\documentclass[acmsmall,screen]{modifiedec}

\copyrightyear{2021}
\acmYear{2021}
\setcopyright{rightsretained}
\acmConference[EC '21] {Proceedings of the 22nd ACM Conference on Economics and Computation}{July 18--23, 2021}{Budapest, Hungary}
\acmBooktitle{Proceedings of the 22nd ACM Conference on Economics and Computation (EC '21), July 18--23, 2021, Budapest, Hungary}
\acmPrice{}
\acmISBN{978-1-4503-8554-1/21/07} \acmDOI{10.1145/3465456.3467622}

\hypersetup{colorlinks=true,allcolors=teal}

\usepackage{booktabs} %
\usepackage[ruled,vlined,linesnumbered]{algorithm2e} %

\SetAlFnt{\small}
\SetAlCapFnt{\small}
\SetAlCapNameFnt{\small}
\SetAlCapHSkip{0pt}
\IncMargin{-\parindent}

\setcitestyle{acmnumeric}

\usepackage[capitalize,nameinlink]{cleveref}
\usepackage{enumitem} %
\setlist{parsep = 0.2\parsep, itemsep = 0.2\itemsep}
\usepackage{bm}
\usepackage{mathtools,thm-restate}
\theoremstyle{plain}
\newtheorem{theorem}{Theorem}
\crefname{theorem}{Theorem}{Theorems}
\newtheorem{lemma}[theorem]{Lemma}
\crefname{lemma}{Lemma}{Lemmas}

\crefname{corollary}{Corollary}{Corollaries}

\crefname{proposition}{Proposition}{Propositions}
\theoremstyle{definition}

\crefname{example}{Example}{Examples}
\crefname{mechanism}{Mechanism}{Mechanisms}
\crefname{procedure}{Procedure}{Procedures}
\crefname{line}{line}{lines}
\DeclareMathOperator*{\argmax}{argmax}

\usepackage{textcase}
\usepackage{etoolbox}

\hypersetup{colorlinks=true,allcolors=teal}

\makeatletter
\patchcmd{\@sect}{\uppercase}{\MakeTextUppercase}{}{}
\patchcmd{\@sect}{\uppercase}{\MakeTextUppercase}{}{}
\makeatother

\usepackage{dsfont}

\usepackage{xcolor} %
\usepackage{tikz}
\usetikzlibrary{fit,calc}
\newcommand*{\tikzmk}[1]{\tikz[remember picture,overlay,] \node (#1) {};\ignorespaces}
\newcommand{\boxit}[1]{\tikz[remember picture,overlay]{\node[yshift=3pt,fill=#1,opacity=.25,fit={($(A)+(.005\linewidth, .2\baselineskip)$)($(B)+(.96\linewidth,.8\baselineskip)$)}] {};}\ignorespaces}
\colorlet{highlight}{black!20}

\newcommand{\boldone}{\mathds{1}}

\newcommand{\opt}{\mathit{Opt}}
\newcommand{\sopt}{\mathit{SOpt}}

\newcommand{\avgopt}{\mathit{AvgOpt}}
\newcommand{\totlen}[1]{||#1||}
\newcommand{\paths}{\mathit{paths}}
\newcommand{\Paths}{\mathit{Paths}}
\newcommand{\upaths}{\bigcup_{i}\Paths_i}
\newcommand{\Explored}{\mathit{Exp}}
\newcommand{\nxt}{\nu}
\newcommand{\opaths}{\mathit{opaths}}
\newcommand{\Benchmark}{\mathit{BM}}
\newcommand{\benchmark}{\mathit{bm}}
\newcommand{\irpath}{\pi_{\mathit{IR}}}

\makeatletter

\newenvironment{customthm}[1]
  {\count@\c@theorem
   \global\c@theorem#1 %
    \global\advance\c@theorem\m@ne
   \theorem}
  {\endtheorem
   \global\c@theorem\count@}

\makeatother

\title{Incentive-Compatible Kidney Exchange in a Slightly Semi-Random Model}
\author{Avrim Blum}
\affiliation{\institution{Toyota Technological Institute at Chicago}\country{USA}}
\author{Paul G\"olz}
\affiliation{\institution{Carnegie Mellon University}\country{USA}}

\begin{abstract}
Motivated by kidney exchange, we study the following mechanism-design problem:
On a directed graph (of transplant compatibilities among patient--donor pairs), the mechanism must select a simple path (a chain of transplantations) starting at a distinguished vertex (an altruistic donor) such that the total length of this path is as large as possible (a maximum number of patients receive a kidney).
However, the mechanism does not have direct access to the graph.
Instead, the vertices are partitioned over multiple players (hospitals), and each player reports a subset of her vertices to the mechanism.
In particular, a player may strategically omit vertices to increase how many of her vertices lie on the path returned by the mechanism.

Our objective is to find mechanisms that limit incentives for such manipulation while producing long paths.
Unfortunately, in worst-case instances, competing with the overall longest path is impossible while incentivizing (approximate) truthfulness, i.e., requiring that hiding nodes cannot increase a player's utility by more than a factor of $1 + o(1)$. %
We therefore adopt a semi-random model where a small ($o(n)$) number of random edges are added to worst-case instances.
While it remains impossible for truthful mechanisms to compete with the overall longest path, we give a truthful mechanism that competes with a weaker but non-trivial benchmark: the length of any path whose subpaths within each player have a minimum average length.
In fact, our mechanism satisfies even a stronger notion of truthfulness, which we call \emph{matching-time incentive compatibility}.
This notion of truthfulness requires that each player not only reports her nodes truthfully but also does not stop the returned path at any of her nodes in order to divert it to a continuation inside her own subgraph.
\end{abstract}

\setcopyright{rightsretained}


\begin{document}

\begin{titlepage}
\maketitle
\end{titlepage}

\section{Introduction}
\label{sec:intro}
For many years, the patients in need of a kidney transplant have far outnumbered the organs available from deceased donors~\cite{MHC18}.
While many patients can find a friend or relative willing to donate a live organ to them,
medical incompatibilities frequently prevent such a direct donation.
Kidney-exchange platforms address this problem by matching patients to kidneys across multiple patient--donor pairs.
The overarching goal in designing these platforms is to maximize the number of patients who receive transplants, i.e., the \emph{welfare} of the matching.
The input to the mechanism can be represented as a \emph{compatibility graph}, a directed graph whose nodes represent patient--donor pairs, and where an edge indicates that the first pair's donor can donate to the second pair's patient.

Traditionally, research on kidney exchange has focused on matchings\footnote{As is customary in the kidney-exchange literature, we use the word ``matching'' differently from most of computer science. A matching is a set of edges in a directed graph such that each node has indegree and outdegree at most one, and for each node (except for the altruist, which we define later), the outdegree is at most the indegree.} consisting of disjoint 2-cycles and 3-cycles.
Cycles are a natural choice because a paired donor can only be expected to donate a kidney if her patient receives one; the short length of cycles allows surgeries to proceed simultaneously and thus prevents donors from reneging.
In this body of work, compatibility graphs are often modeled by classes of densely connected random graphs.
If these random graphs are sufficiently large, 2- and 3-cycles can match essentially all pairs that one could hope to match~\cite{AR11}.
However, these properties of the random model do not reflect the observation that, in real kidney-exchange platforms, more complex matchings can match substantially more patient--donor pairs~\cite{AGR+12}.

Indeed, subsequent work~\cite{AGR+12,DPS12} has highlighted that welfare can be increased substantially by one such form of matchings in particular: long chains starting at non-directed altruistic donors.
These donors, which we simply refer to as \emph{altruists} in the rest of the paper, are donors willing to donate a kidney without requiring a transplant in exchange.
When an altruist donates a kidney, a compatible patient can receive a kidney \emph{before} her donor donates to another patient.
Thus, even if one of the donors should renege, no pair will have donated without having received a kidney in return, which means that simultaneous transplantations are no longer required.
Since donors still renege very infrequently~\cite{CGN+17}, an altruistic donor can in principle be at the beginning of an arbitrary-length chain of donations.\footnote{As of February 2021, \href{https://web.archive.org/web/20210209092241/https://www.uab.edu/kidneychain/}{one particular chain} has led to 114 successful transplants and is still ongoing.}

In this paper, we study the problem of finding such a long path starting from a given altruist.
For this path to obtain high welfare, the path must usually extend across patient--donor pairs in multiple hospitals.
However, finding such a path across hospitals is complicated by the fact that the centralized path-finding mechanism cannot directly observe the patient--donor pairs at each hospital and must rely on hospitals to accurately report them.
Since hospitals feel primarily bound to their own patients' well-being, a hospital might choose not to register certain pairs if this will lead the platform to match more of its other pairs~\cite{AR12}.
Thus, the problem of finding long paths in kidney exchange is one of mechanism design, which means that we search for mechanisms that incentivize hospitals to truthfully report all their clients, which is known as \emph{incentive compatibility}.\footnote{For ease of exposition, we assume that hospitals can hide nodes, not individual edges. However, as discussed in \cref{sec:assumptions}, our results generalize to the case where edges can also be hidden.}

In fact, we propose the stronger notion of \emph{matching-time incentive compatibility}, in which hospitals have additional opportunities for manipulation after the mechanism has determined a path.
At any node on the path belonging to a hospital, the hospital may receive the kidney and then use it to divert the original path to a continuation within its own subgraph.
Matching-time incentive compatibility requires that hiding vertices, diverting the path, or a combination of both does not pay off.
We believe this stronger notion of incentive compatibility, which generalizes the widely-used incentive property of individual rationality, to be of independent interest.
Throughout this paper, we consider $\big(1 + o(1)\big)$-approximate versions of these incentive properties.
That is, if manipulation can increase a hospital's utility by at most a factor of $\big(1+ o(1)\big)$, we assume that hospitals can be convinced to participate truthfully, which is in the common interest.

On worst-case graphs, simple counterexamples show that welfare maximization and incentive compatibility are incompatible (\cref{sec:irworstcase}).
This impossibility holds even for generous approximate notions of welfare maximization and incentive compatibility.
On the other extreme, if the compatibility graphs were purely random Erd\H{o}s-R\'enyi graphs, the problem would become easy for most parameter settings: either there are no long paths or there is a Hamiltonian cycle.
Unfortunately, in practice, kidney-exchange compatibility graphs do not resemble Erd\H{o}s-R\'enyi graphs due to issues such as sensitization and the dynamic nature of how patients enter and exit the system.

Adopting a semi-random approach, we show that one can obtain interesting positive results by adding a small number of random edges to worst-case graphs.
Notably, edges may be added with probability $1/n^{c}$ for $1 < c < 2$, meaning that only a vanishing fraction of vertices will be the endpoint of a random edge.
To our knowledge, these are the first positive results for semi-random graphs with edge probabilities in $o(1/n)$.

In contrast to purely random models, our semi-random model captures a wide range of compatibility graphs, including very sparse and heterogeneous ones.
While the random edges do not directly model a clinical phenomenon,\footnote{One medical justification for these random edges could be that some exchanges perform blood-type incompatible transplantations when tissue-type profiles are well compatible~\cite{FHC+13}. Since we insert random edges only for a $o(1)$ fraction of agents, such phenomena need not be frequent to support our model.} we interpret them as a mild regularity condition on the graph, expressing that a given path in one hospital's subgraph and a given path in another hospital's subgraph are likely to be linked by at least one edge.
Since these random edges are only assumed to exist between pairs of different hospitals, the hospital subgraphs themselves can be purely worst case, somewhat like in the semi-random model of Makarychev, Makarychev, and Vijayaraghavan~\cite{MMV12}.

\subsection{Our Techniques and Results}
\label{sec:techniques}
Even in our semi-random model, maximum welfare is at odds with incentive compatibility.
This is shown by the family of instances in \cref{sec:irnoopt}, in which it is impossible to compete with the optimal path without creating incentives for misrepresentation.
As we discuss in \cref{sec:lowerdiscussion}, these instances are problematic because any high-welfare path must visit one hospital's subgraph many times but spend few steps there each time, a scenario which incentivizes hospitals to hide outgoing edges to other hospitals.
It is natural, then, that we can overcome this impossibility result by benchmarking our mechanism against the longest path such that, for each hospital, the subpaths the hospital's subgraph have some minimum average length.

\begin{figure}
    \centering
    \includegraphics[width=.65\textwidth]{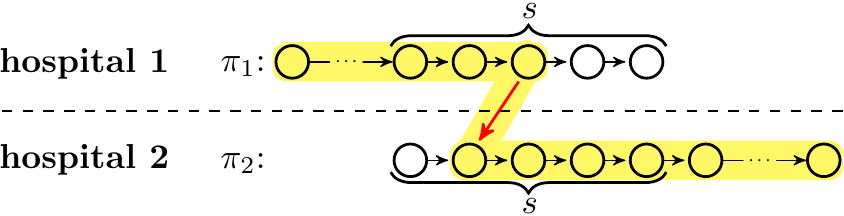}
    \caption{Illustration for stitching together two paths $\pi_1$ and $\pi_2$. If the red edge exists, we can stitch the paths together to obtain the yellow path.}
    \label{fig:stitch}
\end{figure}

The random edges in our model are so sparse that they, by themselves, do not form paths of significant length.
The main insight, however, is that they allow us to ``stitch together'' sufficiently long paths of worst-case edges inside different hospitals.
Indeed, suppose that $p$ is the probability of a random edge existing, and that $s$ is a parameter such that $s^2 \, p \gg 1$.
As illustrated in \cref{fig:stitch}, let there be two paths $\pi_1, \pi_2$ of length greater than $s$, each internal to a different hospital.
Then, the expected number of edges from the the $s$ last nodes of $\pi_1$ to the first $s$ nodes of $\pi_2$ is high.
If any of these edges do exist, we can use them to ``stitch'' $\pi_1$ and $\pi_2$ together.
Then, nearly all nodes of $\pi_1$ and $\pi_2$ will lie on the stitched path except at most $2 \, s$ nodes in the suffix of $\pi_1$ and the prefix of $\pi_2$.

Roughly speaking, our mechanism has the following structure:
From the graph reported by the hospitals, the mechanism chooses a set of disjoint paths within each hospital so that total length is large, and so that truthful reporting maximizes the total length of a hospital's chosen paths.
Then, the mechanism stitches these internal paths together, which will succeed with high probability.
Successfully stitched paths preserve the welfare and hospital utility except for small stitching losses, which gives approximate guarantees on welfare and incentive compatibility.

Since our mechanism can stitch long enough paths with high probability, it can select the internal paths of hospital~$i$ (mostly) based on the subgraph reported by hospital~$i$. As a result, each hospital wants to reveal all its nodes to increase the total length of its chosen paths.
If, by contrast, the mechanism chose its path by how well nodes are connected via cross edges, hospitals might hide connections to other hospitals to sway the mechanism to match more of their own nodes instead. 

Challenges in translating this high-level structure into concrete mechanisms include:
\begin{description}[style=unboxed,leftmargin=0.62cm]
    \item[The choice of original paths.] The paths within each hospital initially selected by the mechanism must compete with the benchmark in terms of welfare, and must lend themselves to stitching.
    \item[The number of paths per hospital.] Stitching only works if we alternate the hospitals on a chain, so we cannot stitch too many paths by a single hospital.
    \item[Internal paths too short for stitching.] The altruist might only be able to reach long paths via an initial sequence of short internal paths connected by worst-case cross edges. In this case, we do have to consider worst-case cross edges in deciding where to stitch while also ensuring that incentives are limited. 
\end{description}%
Our main result is the following:
\begin{customthm}{8}
    Let there be at least two hospitals.
    Let $s, p$ be parameters varying in $n$ such that $p \, s^2 \in \omega\big(\log (n)^5\big)$.
    Let $\avgopt^s$ be the longest path from the altruist such that, for each hospital, the average length of the subpaths in this hospital's subgraph is at least $s$ (or the hospital does not own any nodes on the path).
    
    Then, there is a mechanism that, with high probability, produces a path of length at least $\big(1 - o(1)\big) \, |\avgopt^s|$ assuming truthful responses, and that ensures that each hospital can only increase its utility by lower-order terms through manipulation, including matching-time manipulation.
\end{customthm}
The rest of this paper is organized as follows:
In \cref{sec:relatedwork}, we discuss related work on kidney exchange and semi-random models.
In \cref{sec:lower}, we present lower bounds on the welfare guarantees achievable by incentive-compatible mechanisms, both in the worst case and in our semi-random model.
We also explain why our welfare benchmark restricting the average length of subgraphs for each hospital is a natural way to circumvent these impossibilities.
In \cref{sec:s}, we present a simple mechanism that gives lower guarantees in terms of welfare and incentive properties, and also displays key ideas of our full mechanism, which we present in \cref{sec:avg}.
\Cref{sec:assumptions} discusses alternative modeling assumptions and how they would affect the problem. We conclude in \cref{sec:discussion}.

\subsection{Related Work}
\label{sec:relatedwork}
There is a rich body of research on algorithm and mechanism design for kidney exchange~\cite{BMR09,AR11,TP15,AFK+15,HDH+15,BCH+17}, much of which we have already mentioned.
Even though an influential paper by Ashlagi and Roth~\cite{AR12} specifically identified hospital incentives and altruist-initiated chains as the two main challenges in increasing the welfare of kidney exchange in practice, we know of no existing positive results accounting for both aspects, which is perhaps understandable in light of the robustness of worst-case results as those in \cref{sec:irworstcase}.\footnote{While one paper~\cite{HDH+15} claims to overcome these impossibilities in a repeated setting, there are concerns about the validity of their results, which we have raised with the authors.}
If matchings are restricted to unions of 2-cycles, Ashlagi, Fischer, Kash, and Procaccia~\cite{AFK+15} give an incentive-compatible mechanism with good welfare guarantees, but their approach does not extend to richer matchings.

Whenever a problem is intractable for worst-case inputs and specific input distributions are hard to justify, semi-random models can allow for positive results while retaining substantial generality.
In most works that apply semi-random models, the motivating intractability is computational complexity.
For instance, Blum and Spencer~\cite{BS95} pioneered semi-random models in the context of $3$-coloring, which is NP hard in the worst case.
One of their models starts from an adversarially chosen $3$-colored graph and then, independently for each pair of differently colored nodes, toggles the existence of an edge between the two nodes with some small probability $p$.
If $p \geq n^{-0.6 + \epsilon}$ (and all colors have $\Omega(n)$ nodes), an efficient algorithm 3-colors the resulting graph with high probability.
Since then, many authors have investigated other computational problems in related models~\cite{FK00,KMM11,FK98}; we refer the reader to Feige~\cite{Feige20} for a survey of the literature.
Semi-random models are also used in smoothed analysis~\cite{ST04,ERV14}, again driven by complexity considerations.
As in our work, one can understand positive results as saying that worst-case instances have to be constructed in a very brittle way, which hopefully makes them unlikely to appear in practice.

Only in few cases have semi-random models been applied for desiderata other than complexity.
The most prominent example is the use of random-order models for online algorithms~\cite{GS20,KMT11,Kenyon96}, where a random arrival order makes the problem more tractable in an information-theoretic sense.
In certain versions of the secretary problem, for instance, an adversary chooses a set of quality scores, but the order in which the algorithm sees these scores is chosen uniformly at random.
Here, assuming a random arrival order is not a tool for easier computation; instead, it increases the best approximation ratio of \emph{any} algorithm from $1/n$ to $1/e$.
Babaioff, Immorlica, and Kleinberg~\cite{BIK07} combine these ideas with mechanism design in order to design auctions that are incentive compatible when a set of buyers with worst-case valuations arrive in random order.
The fact that monetary transfers are available in this setting gives their mechanisms a flavor decidedly different from ours.

We are not aware of existing work on semi-random models that gives positive results with edge probabilities as low as ours.
As two examples of work giving relatively low edge probabilities, Makarychev et al.~\cite{MMV12} find planted balanced cuts for edge probabilities in $\mathit{polylog}(n)/n$ for balanced cuts; Chen, Sanghavi, and Xu~\cite{CSX12} recover clusters for a similar range of probabilities in a different semi-random model.
By contrast, our approach can handle edge probabilities as low as $1/n^{2-\epsilon}$.

Our work is most closely related to that of Blum et al.~\cite{BCH+17}, who design mechanisms for kidney exchange in a different semi-random model.
Their model is semi-random in that, in a worst-case compatibility graph, the nodes are allocated randomly to the hospitals.
For matchings composed of disjoint constant-length cycles, their mechanism guarantees optimal welfare and individual rationality up to lower-order terms, a much weaker notion than matching-time incentive compatibility.
Blum et al.\ acknowledge that, unfortunately, their approach ``does not seem to extend beyond individual rationality,'' and they provide an impossibility result showing that individual rationality cannot be obtained if long paths from altruistic donors are allowed.

An argument closely related to the stitching technique described in \cref{sec:techniques} was used by Dickerson and Sandholm~\cite{DS17} in the proof of their Proposition~2.
However, we make much heavier use of this technique: Whereas Dickerson and Sandholm perform a single stitching between two specified subgraphs, we construct a long path consisting of many stitchings across multiple subgraphs.

\section{Model}
\label{sec:model}
Let there be a finite, directed base graph $G_\mathit{base} = (V, E_\mathit{base})$, where $n \coloneqq |V|$.
Let the vertices $V$ be partitioned into subsets $V_i$ for each of finitely many hospitals $i$.
Let $E_p$ denote the random variable ranging over subsets of $\bigcup_{i \neq j} (V_i \times V_j)$, where each potential edge is included in $E_p$ independently with some given \emph{edge probability} $p$.
While our results hold for a wider range of $p$, we are primarily interested in very low values of $p$, which might, for example, scale as $1/n^c$ for $1 < c < 2$.
The \emph{compatibility graph} is defined by adding the random edges to the base graph, i.e., as $G \coloneqq (V, E_{\mathit{base}} \cup E_{p})$.
Let one node in $G$ be labeled as the altruist $\alpha$.
The hospital owning the altruist node is called the \emph{altruist owner} and is denoted by $a$.
We consider an \emph{instance} to be the base graph together with node ownership, the value of $p$, and the identity of the altruist.

Each hospital~$i$ reports a subset $V_i' \subseteq V_i$ to the mechanism.
If $V_i' = V_i$, we say that $i$ reports \emph{truthfully}.
The mechanism then learns about all edges in the subgraph of $G$ induced by $\bigcup V_i'$, which means that the mechanism knows this subgraph, node ownership and the identity of the altruist,\footnote{For well-definedness, we assume that the altruist is always reported ($\alpha \in V_a'$). This is without loss of generality: If the altruist is hidden, the mechanism cannot output anything and, thus, no hospital can receive more utility from the mechanism. The more promising strategy of using the altruist for internal matching rather than bringing it to the mechanism will be captured by the definition of matching-time incentive compatibility even without hiding nodes.} but does not have access to the remaining nodes and edges, to whether edges are random or deterministic, and to the value of $p$.
Based on this knowledge, the mechanism selects a (simple) path starting at the altruist.

Let the \emph{length} $|\pi|$ of a path $\pi$ be its number of nodes, and let the \emph{total length} $\totlen{P}$ of a set of vertex-disjoint paths $P$ be the sum of their lengths.
We say that paths are \emph{disjoint} when they are vertex disjoint.
We refer to the length of the path produced by the mechanism as the \emph{welfare}, and to the number of vertices in this path that belong to a certain hospital as the \emph{utility} of this hospital.

We say that edges between nodes of the same hospital are \emph{internal edges}; all other edges are \emph{cross edges}.
When cross edges are considered as part of a path, we call them \emph{hops}.
A path entirely consisting of internal edges of hospital~$i$ is called an \emph{internal} or \emph{$i$-internal} path, and we refer to the maximal internal subpaths of a path as the \emph{segments} of the path.

The welfare of a mechanism can be measured relative to multiple benchmarks, all of whom we define on the full compatibility graph (rather than just the base graph).
We define our strongest benchmark $\bm{\opt}$ as the longest path starting at the altruist.
We also define two weaker benchmarks.
First, for a given $s$, let the \emph{long-segments benchmark}, $\bm{\sopt^s}$, be the longest path from the altruist such that every segment has at least length $s$, i.e., such that no hop appears too close to another hop, or to the start or end of the path.
Second, for some $s$, let the \emph{high-averages benchmark}, $\bm{\avgopt^s}$, be the longest path from the altruist such that each hospital has either no segments in the path, or its mean segment length is at least $s$.
Either of the two weaker benchmarks might not be defined; in these cases, consider the benchmark as the empty path.
Note that, on every compatibility graph, the total length of $\avgopt^s$ and $\sopt^s$ is monotone nonincreasing in $s$, and that, for every compatibility graph and for every $s$, $\totlen{\sopt^s} \leq \totlen{\avgopt^s} \leq \totlen{\opt}$.

We search for a mechanism that, with high probability over the random edges of a given instance, satisifies the following desiderata:
\begin{description}[style=unboxed,leftmargin=0.62cm]
    \item[Efficiency:] If all hospitals report truthfully, the welfare of the mechanism is at least a $\big(1 - o(1)\big)$ fraction of the welfare of a given benchmark.
    \item[Incentive Compatibility:] If all hospitals report  truthfully except for hospital~$i$, $i$'s utility from the path produced by the mechanism is at most a $\big(1 + o(1)\big)$ fraction of its utility when all hospitals report truthfully.
    \item[Individual Rationality:] Under truthful reports, the utility of the altruist owner is always at least a $\big(1 - o(1)\big)$ fraction of the length of its longest internal path $\irpath$ starting at the altruist.
    \item[Matching-Time Incentive Compatibility:] As for regular incentive compatibility, let $\pi$ be the path produced by the mechanism when all hospitals except for possibly $i$ report truthfully.
    Not only must $i$'s utility from $\pi$ be at most a $\big( 1 + o(1)\big)$ fraction of its truthful utility, but the same must be true for any path $\pi'$ created by diverting $\pi$ at one of $i$'s nodes $v$ into an internal continuation.\footnote{Formally: For some node $v \in V_i$ on $\pi$, write $\pi = \pi_1 \, v \, \pi_2$ for the prefix and suffix, and let $\pi_\mathit{ext}$ be an internal path of hospital~$i$, starting at $v$, disjoint from $\pi_1$, and not necessarily contained in the nodes reported by $i$. Then, $\pi' = \pi_1 \pi_{\mathit{ext}}$.}
\end{description}
Note that matching-time incentive compatibility implies incentive compatibility.
Matching-time incentive compatibility also generalizes individual rationality since the altruist owner may obtain her longest internal path $\irpath$ starting at the altruist by reporting all her nodes and diverting the path at the altruist into $\irpath$.

In \cref{sec:assumptions}, we discuss the effects of alternative modeling choices on our problem.

\section{Lower Bounds}
\label{sec:lower}
In this section, we give lower bounds on the welfare of all mechanisms that are either approximately individually rational or approximately incentive compatible.
For brevity, we present the arguments for individual rationality here and defer the proofs for incentive compatibility to \cref{app:lower}.

Consider any mechanism in which, with high probability, the altruist owner can increase her utility by at most $o(n)$ by dropping out and matching internally.\footnote{Note that our mechanisms fall into this class of mechanisms because, as we will show, hospitals can gain at most an $o(1)$ fraction of the utility obtained under truthful reporting by manipulating, and because the utility obtained under truthful reporting is at most $n$.}
If compatibility graphs are worst case with no added random edges, we show in \cref{sec:irworstcase} that the hypothetical mechanism will fall short of even our suboptimal benchmarks (even for some $s \in \Theta(n)$) by a gap in $\Theta(n)$ with constant probability.
In \cref{sec:irnoopt}, we then show that even with randomness, the hypothetical mechanism still cannot compete with the welfare of $\opt$, again by a gap in $\Theta(n)$ with constant probability.

The second of these lower bounds illustrates why our positive results can compete with the suboptimal benchmarks $\sopt^s$ and $\avgopt^s$, but not the full $\opt$.
More subtly, these analyses also show that purely efficiency-maximizing mechanisms can create large incentives for manipulation, resulting in hospitals being unlikely to participate truthfully or at all.
We show in \cref{sec:lowerdiscussion} that, in instances such as those in \cref{thm:irnoopt}, our mechanism actually gives higher utility to each hospital than the pure welfare maximization after accounting for hospital behavior due to incentives.
Finally, we explain how the lower bound in \cref{thm:irnoopt} guided our choice of suboptimal benchmarks.

\subsection{Worst-case graphs}
\label{sec:irworstcase}
\begin{theorem}
    \label{thm:irworstcase}
    Let $p=0$, i.e., suppose that no random edges are added to the worst-case base graph.
    For each deterministic mechanism, there exists a family of instances such that, on every instance, the mechanism will obtain $\Theta(n)$ less welfare than $\sopt^{n/4}$ or violate individual rationality by $\Theta(n)$.
    (Note that, if the mechanism falls short of $\sopt^{n/4}$, it also falls short of the stronger benchmarks $\sopt^s$ and $\avgopt^s$ for any $s \leq n/4$.)
\end{theorem}
\begin{proof}
For any $k \in \mathbb{N}_{\geq 1}$, consider the instance given in \cref{fig:worstcase}.
Clearly, $\sopt^{n/4}$ will select the path of length $3 \, k = 3/4 \, n$.
If the mechanism selects a path of length at most $2 \, k$, it falls short of the optimal welfare by at least $k = n/4$.
Else, the path must visit some of hospital~2's nodes, and therefore can visit at most $k$ nodes in hospital~1.
However, hospital~1 is the altruist owner and can match all her $2 \, k$ nodes internally outside of the mechanism.
Thus, the mechanism must either fall short of the welfare of $\sopt^{n/4}$ by $n/4$, or it must give hospital~1 an incentive of at least $n/4$ to not participate in the mechanism.
\end{proof}
\Cref{thm:icworstcase} in \cref{app:lower} shows an analogous result for incentive compatibility (rather than individual rationality) with respect to $\sopt^{n/6}$.

At a high level, the graphs that give rise to these results are problematic because they force the mechanism to choose between a branch with higher welfare and a branch that satisfies key individual hospitals (in \cref{thm:irworstcase}, the altruist owner), who, if unsatisfied, could remove the high-welfare option altogether by not participating or hiding critical nodes.
Intuitively, our semi-random model gives positive results because, when $p$ is large enough, there are enough non-adversarial cross-edges for the mechanism to avoid these problematic decision points in the graph.
For example, in the instance in \cref{fig:worstcase}, if random edges allow for the stitching of hospital~1's $2\,k$-length path in front of the path of hospital~2, then obtaining high welfare and satisfying hospital~1 become compatible.

\subsection{Competing with Optimal Welfare in the Semi-Random Model}
Unfortunately, the low number of random edges we allow for in our model do not entirely prevent these incompatibilities between incentives and welfare.
In particular, within our model, we can still force the mechanism to make these difficult choices on a ``small scale''\,---\,that is, between segments that are so short that we likely cannot stitch them using random edges.
These small-scale choices become a problem for our guarantees in instances where the mechanism must make them repeatedly.
We use this observation in the proof of \cref{thm:irnoopt} to construct instances in which our mechanism cannot compete with $\opt$ while maintaining individual rationality. The corresponding result for incentive compatibility is shown by \Cref{thm:icnoopt} in \cref{app:lower}.
\label{sec:irnoopt}
\begin{theorem}
    \label{thm:irnoopt}
    Fix any $p \in o(1/n)$ and any deterministic mechanism.
    Then, there exists a family of instances such that, with $\Omega(1)$ probability, the mechanism will on each instance either
    obtain $\Theta(n)$ less welfare than $\opt$ or violate individual rationality by $\Theta(n)$ .
\end{theorem}
\begin{figure}[!tb]
\centering
\begin{minipage}{.38\textwidth}
  \centering
    
    \includegraphics[height=.49\textwidth]{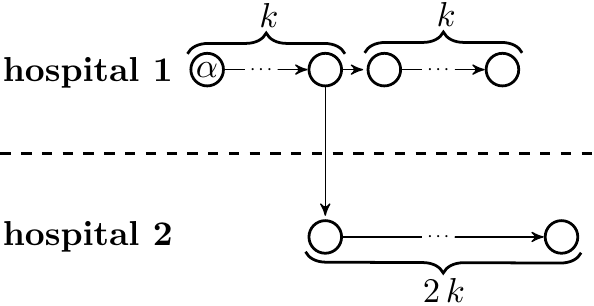}
    \caption{Without random edges, individually rational mechanisms cannot compete with $\sopt^{n/4}$ and $\avgopt^{n/4}$\!.}
    \label{fig:worstcase}
\end{minipage}\hspace{.04\textwidth}%
\begin{minipage}{.58\textwidth}
  \centering
    \includegraphics[width=\textwidth]{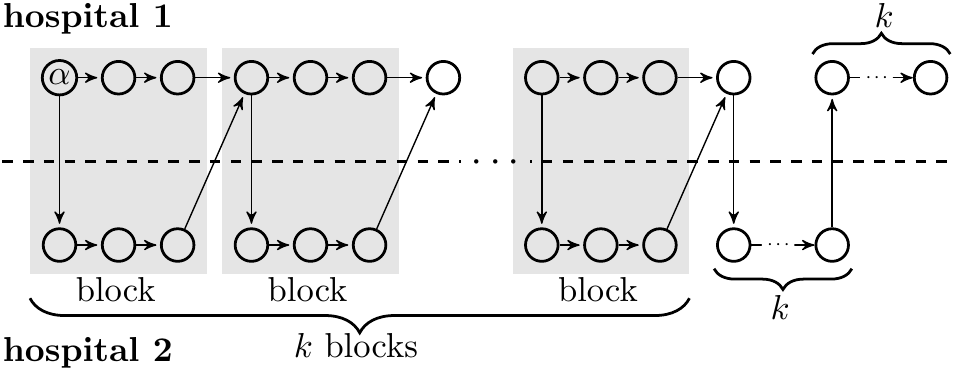}
    \caption{Even with random edges, individually rational mechanisms cannot compete with maximum welfare.}
    \label{fig:noopt}
\end{minipage}
\end{figure}
\begin{proof}
For any $k \in \mathbb{N}_{\geq 1}$, consider the base graph of size $8 \, k + 1$ in \cref{fig:noopt}.
Even relying only on deterministic edges, there is always a path of length $6 \, k + 1$, which can be obtained by branching into hospital~2's subgraph in every block ($k$ times one node from hospital~1 and three nodes from hospital~2) and going through both segments of length $k$ at the end.

Observe that the number of random edges added to this base graph is binomially distributed with success probability $p$ and at most $n^2$ trials.
With probability $1/2$, the number of random edges is at most its rounded expectation, which is less than $p \, n^2 + 1 \in o(n)$.
We therefore condition on the assumption that the number of random edges is at most $p \, n^2 +1$ for the remainder of the proof.

If the mechanism matches at most $11/4 \, k$ of hospital~1's nodes with probability at least $1/2$, then with total probability $1/4$ there is an incentive of at least $3 \, k + 1 - 11/4 \, k = k/4 + 1 \geq  n/32$ for hospital~1 to match outside of the mechanism, concluding the proof of this case.

Else, we additionally condition on the event of more than $11/4 \, k$ of hospital~1's nodes being matched.
At least $7/4 \, k$ of these nodes must lie within the blocks.
For the sake of contradiction, assume that the mechanism obtains welfare of at least $(6 - 1/4) \, k + 1= 23/4 \, k + 1$.
Then, at least $15/4 \, k$ nodes must be matched within the blocks.
Since there are $o(n)$ random edges in total, at most $o(n)$ blocks have some incoming random edge.
Call the number of these blocks $b$.
In each of the $k - b$ remaining blocks, the mechanism has to decide between the upper and lower branch, i.e., it cannot take more than one of hospital~1's nodes while taking more than $3$ nodes in total.
Let $\ell$ be the number of blocks in which the upper branch is chosen.
Since at least $7/4 \, k$ of hospital~1's nodes are matched, it must hold that $(b + \ell) \, 3 + (k - b - \ell) \, 1 \geq 7/4 \, k$, and thus that $\ell \geq 3/8 \, k - b$.
Since at least $15/4 \, k$ nodes are matched in total, it must hold that $b \, 6 + \ell \, 3 + (k - b - \ell) \, 4 \geq 15/4 \, k$ and thus that $\ell \leq 1/4 \, k + 2 \, b$.
Since $b \in o(n) = o(k)$, for large $n$, these inequalities are in contradiction.
We conclude that, with total probability at least $1/4$, welfare is at least $k/4$ lower than $6 \, k + 1$, which is a lower bound on the length of $\opt$.
\end{proof}

\subsection{Discussion of Lower Bounds}
\label{sec:lowerdiscussion}
In the previous two sections, we have shown that both individual rationality and incentive compatibility imply that a mechanism has to match substantially fewer nodes than the longest feasible path.
It is natural to wonder whether these incentive properties are worth the loss in welfare.

However, if strong incentives push hospitals towards deviation, the cost to welfare can be significantly higher.
Indeed, the example in \cref{fig:noopt} illustrates that maximizing welfare without concern for incentives can lead to lower welfare than is given by our mechanisms.
By reasoning similar to that used in \cref{thm:irnoopt}, a welfare-maximizing mechanism will (with high probability) give the altruist owner $\Theta(n)$ incentive to match outside of the mechanism, giving the altruist owner utility $3 \, k+1$ and hospital~2 utility $0$.
By contrast, one can verify that, for $s \leq (n-1)/8$, both mechanisms that we develop in this paper (with high probability) choose a path giving the altruist owner utility in $\big(4 - o(1)\big) \, k$ and hospital~2 utility in $\big(1 - o(1)\big) \, k$.
Except for $o(n)$ nodes more or fewer, this path will visit all of hospital~1's nodes in the blocks, and all the nodes in the $k$-length segments of both hospitals.
Thus, using our mechanisms rather than na\"ive welfare maximization benefits both hospitals in this situation.

In the rest of this section, we explain why we consider the long-segments and high-averages benchmarks to be natural responses to the lower bound in \cref{thm:irnoopt}.
We start with the simple observation that, in our semi-random model, it is advantageous for a hospital to have long internal paths.
This is because long internal paths are more likely to have cross-edges that are ``useful'' to the mechanism\,---\,that is, cross-edges that the mechanism could use to increase overall welfare.
When a hospital has useful cross-edges to and from its internal paths, the mechanism will incorporate many of the hospital's nodes into a high-welfare path, an outcome which gives utility to the hospital while increasing overall welfare, thereby aligning the objectives of this hospital and the mechanism.

Unfortunately, long internal paths do not guarantee a hospital high utility in the mechanism, because a hospital's long paths give it utility only to the extent that the mechanism can use them to create welfare.
This problem arises in the example in \cref{fig:noopt}, in which hospital~1 has a long internal path, but yet ends up getting $\Omega(n)$ less utility than it would receive by leaving the mechanism altogether.
What happens to hospital~1 in this example illustrates the broader problem addressed by our benchmarks:
to even approximate maximum welfare, the mechanism's path must use many short segments of hospital~1's internal path, leaving the remainder chopped into pieces that are (with high probability) too short to be stitched in later.
For each of these short segments, then, hospital~1 gains little utility (one node) while losing a comparable amount of potential utility (2 nodes).
By this reasoning, any path with welfare close to $\opt$ is unacceptable to hospital~1, which would be better off matching its internal path outside of the mechanism.

Our benchmark $\sopt^s$ addresses this problem by enforcing that, for each hospital, all their segments in the path have length at least $s$.
We will choose $s$ to be much larger than the maximum amount of nodes on a long path that might become useless for stitching due to a single segment.
By the reasoning above, this restriction guarantees that the utility each hospital gains from segments included in the mechanism's path outweighs the potential loss in stitchable paths.
The high-averages benchmark $\avgopt^s$ is based on a slightly more refined version of the argument for $\sopt^s$, which recognizes the fact that it suffices if the utility from \emph{all} segments outweighs the losses afflicted by \emph{all} segments, rather than each segment paying off individually.
Obviously, there exist base graphs on which obtaining more utility than our benchmarks would not be in conflict with incentives.
However, our benchmarks allow us to make a welfare guarantee for all base graphs, which should be high for graphs with long internal paths and ensures that the mechanism makes use of a wide range of opportunities for matching between hospitals.

\section{A Simple Incentive-Compatible Mechanism}
\label{sec:s}
It remains to show that we indeed can compete with the high-averages benchmark while providing matching-time incentive compatibility.
To build intuition for the full result, we begin with a simple mechanism, which only satisfies traditional incentive compatibility (but not individual rationality and matching-time incentive compatibility) and only competes with the long-segments benchmark.

Designing mechanisms would be fairly easy if our graphs were Erd\H{o}s-Rényi style graphs with high edge probability, because the mechanism could simply construct a path including essentially all nodes, and no hospital would have an incentive to deviate.
Unfortunately, the low probability with which our random edges appear prevents us from applying this idea.
Instead, suppose that, for each hospital~$i$, we fix a set $\Paths_i$, which contains disjoint $i$-internal paths that have length much larger than $2 \, s''$ for some $s''$.
Then, consider these paths as the nodes in a meta graph, where the meta graph has an edge between nodes $\pi_1$ and $\pi_2$ iff $G$ has an edge from one of the last $s''$ nodes of path $\pi_1$ to one of the first $s''$ nodes of path $\pi_2$. 
Thus, edges in the meta graph exist whenever at least one of the $s''^2$ potential random edges between $\pi_1$ and $\pi_2$ exists in $G$.
These edges occur independently between meta edges and with probability $1 - (1 - p)^{s''^2}$.
Therefore, if $s''$ is large enough, the meta graph will be densely connected, and we will be able to stitch all paths in $G$ into a single long path.

\begin{algorithm}[tbp]
\DontPrintSemicolon
\SetKwInOut{Input}{input}
\SetKwInOut{Output}{output}
\SetAlgorithmName{Mechanism}{mechanism}{List of Mechanisms}

\Input{reported graph of size $n \gg 1$, parameter $s \gg \log (n)^2$}
\Output{success indicator, path $\pi$}
$s' \leftarrow s / \log n$, $s'' \leftarrow s / \log (n)^2$\;\label{lin:s1:1}
\If{$\nexists$ internal path of length $s'$ from $\alpha$}{\label{lin:s1:2}
\Return ``success'', minimum path $\{\alpha\}$
}\label{lin:s1:3}
\ForEach{hospital $i$}{\label{lin:s1:4}
$\Paths_i \leftarrow$ largest set of disjoint $i$-internal paths each of length exactly $s'$; for $i=a$, one path must start at $\alpha$.
}\label{lin:s1:5}
$j \leftarrow \argmax_{i} |\Paths_i|$, break ties by hospital id.\;\label{lin:s1:6}
$\Paths_j \leftarrow$ set of disjoint $j$-internal paths of max.\ total length subject to (1)~the number of paths being at least $\max_{i \neq j} |\Paths_i|$ and at most $\sum_{i \neq j} |\Paths_i| + \boldone\{j=a\}$, (2)~all paths having size $\geq s'$, and (3)~if $j=a$, one path starting at $\alpha$.\;\label{lin:s1:7}
order the paths $\bigcup_{i} \Paths_i$ in a hospital-alternating sequence such that the path starting at $\alpha$ is the first in the sequence.\;\label{lin:s1:8}
check if paths can be stitched in order by finding an edge from the last $s''$ nodes of the $t$th path in the sequence to the first $s''$ nodes of the $(t+1)$th path for all $t$.\;\label{lin:s1:9}
\uIf{this stitching succeeds}{\label{lin:s1:10}\Return ``success'', stitched path $\pi$}\label{lin:s1:11}
\Else{\label{lin:s1:12}\Return ``failure'', minimum path $\{\alpha\}$}\label{lin:s1:13}
\caption{Long-segments mechanism}\label[mechanism]{alg:s1}
\end{algorithm}
The mechanism built from these intuitions is shown in \cref{alg:s1}.
It takes in the reported graph including node ownership and the identity of the altruist and a parameter $s$.
Recall that the altruist owner is denoted by $a$.

Initially, the mechanism defines for each hospital~$i$ a set $\Paths_i$ of disjoint $i$-internal paths of equal length $s' \gg 2 \, s''$ (\cref{lin:s1:5}).
The main complication standing in the way of stitching $\upaths$ is that random edges only allow for the stitching of paths in different hospitals.
Thus, no hospital can contribute more than half (rounded up for $a$) of all paths in $\upaths$.
To deal with this constraint, in \cref{lin:s1:7}, the mechanism redefines $\Paths_j$ for the hospital with the largest number of paths such that the hospital now has at most as many paths as all other hospitals combined (plus one, for $a$).
In \cref{lin:s1:8}, the mechanism orders all paths in $\bigcup_i \Paths_i$ in a sequence such that the path owners alternate.
We will show further below that this is always possible by applying the following combinatorial lemma, whose proof we defer to \cref{sec:lemballs}:
\begin{restatable}{lemma}{lemballs}
    \label{lem:balls}
    Let there be a finite number of balls with colors in $[c]$, such that there are $n_1 \geq n_2 \geq \cdots \geq n_c$ many balls of each color.
    If $n_1 \leq \sum_{i \geq 2} n_i$, the balls can be arranged in a cyclical sequence such that no two adjacent balls have the same color.
    In particular, they can be ordered in a linear, color-alternating sequence such that an arbitrary ball can be chosen as the first.
\end{restatable}
For our choices of parameters, we will show that, with high probability, the meta graph will be so dense that the chosen order of $\upaths$ forms a path in the meta graph, which means that the paths in $\upaths$ can be stitched in order.
Assuming that this stitching succeeds, the stitched path is returned in \cref{lin:s1:11}; in the low-probability event of failure, the mechanism returns in \cref{lin:s1:12}.

We now comment on the well-definedness of \cref{alg:s1}.
In \cref{lin:s1:1}, we ignore the rounding of $s'$ and $s''$.
Since we assume $s$ to be large, this does not influence the asymptotics of our analysis.
The check in \cref{lin:s1:2} ensures that \cref{lin:s1:5} is well-defined.
\Cref{lin:s1:7} is well-defined because an appropriately chosen subset\footnote{That is, if $j \neq a$, any subset of size $\max_{i \neq j} |\Paths_i|$; else, any subset of size $\min\left(1, \max_{i \neq j} |\Paths_i|\right)$ that includes the path starting at $\alpha$.} of the original $\Paths_j$ satisfies the requirements.
For \cref{lin:s1:8}, we need to show that such an ordering exists.
If $|\Paths_j| \leq \sum_{i \neq j} |\Paths_i|$, \cref{lem:balls} directly gives us a sequence.
Else, by \cref{lin:s1:7}, it must be that $j=a$ and $|\Paths_j| = \sum_{i \neq j} |\Paths_i| + 1$.
Then, \cref{lem:balls} allows to order all paths in $\upaths$ except for the one starting at $\alpha$\footnote{If the path starting at $\alpha$ is the only path in $\upaths$, an ordering exists trivially.} such that the sequence starts with a hospital other than the altruist owner.
Then, adding the altruist path to the front of the sequence will produce a valid hospital-alternating sequence.
Finally, in \cref{lin:s1:9}, for large $n$, note that $s' = \log n \, s'' \gg 2 \, s''$, so that the nodes used for stitching into and out of a path do not overlap.

With these things out of the way, our theorem is the following:
\begin{theorem}
\label{thm:s1}
    Let there be at least two hospitals, and let $s, p$ be parameters varying in $n$ such that $p \, s^2 \in \omega\big(\log (n)^5\big)$.
    Then, for any sequence of instances with such $n$ and $p$, running \cref{alg:s1} on the full graph succeeds with high probability.
    In this case, call the resulting path $\pi$.
    It must hold that $|\pi| \geq \big(1 - o(1)\big) \, |\sopt^s|$ and, by hiding vertices, any hospital~$i$ can receive at most utility $\big(1 + o(1)\big) \, |\pi \cap V_i|$ from the mechanism.
\end{theorem}
To prove this theorem, we show in the following lemmas that the mechanism indeed succeeds on the truthful graph with high probability (\cref{thm:s1succeeds}), that it achieves high welfare in this case (\cref{thm:s1welfare}), and that, if the mechanism succeeds on the truthful graph, manipulation can only increase utility by lower-order term (\cref{thm:s1ic}).

\begin{lemma}
\label{thm:s1succeeds}
    If the parameters $p$ and $s$ vary in $n$ such that $p \, s^2 \in \omega\big(\log (n)^5\big)$, \cref{alg:s1} succeeds with high probability under truthful reports.
\end{lemma}
\begin{proof}
    If the branch in \cref{lin:s1:2} is taken, the mechanism succeeds deterministically.
    Else, it will get to \cref{lin:s1:9}.
    Even if we only consider the random edges, each individual stitching fails with probability at most $(1 - p)^{(s'')^2}$.
    Adding the deterministic edges can only decrease this probability.
    Using a union bound over at most $n$ paths, we conclude that the stitching\,---\,and thus the mechanism\,---\,fails with probability at most
    \[ n \, (1 - p)^{s''^2} \leq \exp(\log n - p \, s^2 / \log (n)^4) \to 0. \qedhere \]
\end{proof}

\begin{lemma}
    \label{thm:s1welfare}
    Suppose that $s \in \omega\big(\log (n)^2\big)$.
    Whenever \cref{alg:s1} succeeds on the truthful reports, it achieves welfare of at least $\big(1 - o(1)\big) \, |\sopt^s|$.
\end{lemma}
\begin{proof}
    Assuming that all hospitals are truthful, consider the path $\sopt^s$.
    In \cref{lin:s1:4,lin:s1:5}, every path segment in $\sopt^s$ of length $\ell \geq s$ could be cut up into $\lfloor\ell / s'\rfloor \geq \ell / s' - 1 \geq (1 - 1/\log n) \, \ell / s'$ disjoint paths of length $s'$, which means that $\Paths_i$ must contain at least $(1-1/\log n) \, |\sopt^s \cap V_i|/s'$ many paths.
    Thus, the total length of each set $\Paths_i$ is at least $(1 - 1/\log n) \, |\sopt^s \cap V_i|$.

    Some additional reasoning is needed to show that the same remains true after $\Paths_j$ is redefined in \cref{lin:s1:7}.
    Namely, $\sum_{i \neq j} |\Paths_i| + \boldone\{j=a\}$ could be lower than the number of $s''$-length paths needed to cover a $(1 - 1/\log n)$ fraction of nodes in $j$'s segments in $\sopt^s$.
    Suppose that this is the case, and let $\sopt^s$ contain $r$ many segments in total.
    
    Suppose for now that $j \neq a$.
    Then, hospital $j$ cannot have owned more than $r/2$ of the segments in $\sopt^s$ since the ownership alternates between segments and the altruist owns the first segment.
    By the reasoning of the first paragraph of this proof, the at least $r/2$ many segments belonging to other hospitals mean that the $\Paths_i$ for all $i \neq j$ together contain at least $\log n \, r / 2$ many paths, which, for large $n$, is larger than $r/2$.
    Then, the redefinition of $\Paths_j$ in \cref{lin:s1:5} permits at least $r/2$ paths of \emph{minimum} length $s'$, which allows to cover all nodes in $\sopt^s \cap V_j$.
    
    Else, i.e., if $j = a$, essentially the same reasoning applies.
    Hospital $j$ cannot have owned more than $\lceil r/2 \rceil$ of the segments in $\sopt^s$, and the $\Paths_i$ for $i \neq j$ will together contain at least $\log n \, \lfloor r/2 \rfloor$ paths.
    For large $n$, this number is at least $\lceil r/2 \rceil - 1$, which means that the redefinition of $\Paths_j$ allows for at least $\lceil r/2 \rceil$ paths; again, all nodes in $j$'s segments of $\sopt^s$ could be covered by a set of paths satisfying the constraints in the redefinition of $\Paths_j$.
    Thus, whether $j=a$ or not, the total length of $\Paths_j$ continues to be at least $ (1 - 1/\log n) \, |\sopt^s \cap V_j|$.

    Through the stitching process, the length of each path in $\upaths$ (which has size at at least $s'$) might decrease by up to $2 \, s''$, so each hospital $i$ receives utility of at least $(1 - 2/\log n) \, (1 - 1/\log n) \, |\sopt^s \cap V_i| \geq (1 - 3/\log n) \, |\sopt^s \cap V_i|$ from the returned path.
    Adding up over all hospitals, it follows that the social welfare is at least a $(1 - 3/\log n)$ fraction of $|\sopt^s|$.
\end{proof}

\begin{lemma}
    \label{thm:s1ic}
    Again, suppose that $s \in \omega(\log (n)^2)$.
    Assuming that \cref{alg:s1} succeeds under truthful reports and returns a path $\pi$, each hospital can only gain $\big(1 + o(1)\big) \, |\pi \cap V_i|$ utility by hiding vertices.
\end{lemma}
\begin{proof}
    If the mechanism returns a trivial path in \cref{lin:s1:3} under truthful reporting, hiding nodes cannot change this outcome.
    If a manipulation causes the mechanism to return in lines~\ref{lin:s1:3} or \ref{lin:s1:13}, the mechanism returns a subpath of the truthful path, and no hospital can gain from this manipulation.
    We may thus assume that the mechanism successfully returns in \cref{lin:s1:11} under truthful reports and only consider maniplations that make the mechanism return in \cref{lin:s1:11} as well.
    
    First, consider a hospital $i$ who cannot form enough paths of length $s'$ in \cref{lin:s1:5} to become the special hospital $j$.
    Let $r$ be the number of paths in $\Paths_i$ under truthful reporting.
    Since the mechanism succeeded, truthful reporting guarantees a utility of at least $r \, (s' - 2 \, s'')$ to hospital~$i$.
    Since no manipulation can increase the number of chosen paths in $\Paths_i$ and since only nodes in $\upaths$ can end up in a stitched path, the possible utility is bounded above by $r \, s' = 1/(1 - 2 / \log n) \, r \, (s' - 2 \, s'')$.
    Thus, the benefits of manipulation are of lower order.

    Second, consider a hospital $j$ who, under truthful reporting, will have the highest number of paths.
    Fix the updated $\Paths_j$ under truthful reporting.
    If a manipulation by $j$ does not change the fact that the hospital is chosen as $j$ in \cref{lin:s1:6}, any such manipulation can only reduce the total size of $\Paths_j$ in \cref{lin:s1:7}.
    Again, such a change can only be advantageous in terms of the stitching loss, which is lower-order as above.

    If, however, the manipulation turns this hospital into one of the hospitals $i$ who are not chosen in \cref{lin:s1:6}, the manipulation must lead to $\Paths_i$ as chosen in \cref{lin:s1:5} containing at most $\max_{i' \neq i} |\Paths_i^s|$ paths.
    Thus, $\Paths_i$ contains fewer paths than the updated $\Paths_j$ under truthful reporting, and each path has  smaller length.
    It follows that the total length of $\Paths_i$ decreases in the manipulation, which means that the hospital may at most gain lower-order utility from a more advantageous stitching.
\end{proof}
Note that, when we defined $s'$ as $s/\log n$ and $s''$ as $s/\log n$, any other function $f(n) \in \omega(1)$ could have replaced the logarithm.
By choosing a function $f$ that grows very slowly, the precondition of \cref{thm:s1} can be slightly weakened to $\lim_{n \to \infty} p \, s^2 / f(n)^4 - \log n \to \infty$.

\section{Matching-Time Incentive Compatibility with Higher Welfare}
\label{sec:avg}
\begin{algorithm}[!tb]
\DontPrintSemicolon
\SetKwInOut{Input}{input}
\SetKwInOut{Output}{output}
\SetAlgorithmName{Mechanism}{mechanism}{List of Mechanisms}
\SetKwRepeat{Do}{do \{}{\} while}
\SetKwFor{RepTimes}{repeat}{times}

\Input{reported graph of size $n \gg 1$, parameter $s \gg \log (n)^2$}
\Output{success indicator, path $\pi$}
$s' \leftarrow s / \log (n)^2$\;\label{lin:avg2:1}
\ForEach{hospital $i$}{\label{lin:avg2:2}
    $\paths_i \leftarrow \log n \times$max.\ number of disj.\ $i$-internal paths such that average path length $\geq s$.\;\label{lin:avg2:3}
}
say that hospitals $i$ with $\paths_i > 0$ are \emph{active}.\;\label{lin:avg2:4}
\If{$a$ inactive or if only $a$ is active}{\label{lin:avg2:5}
    \Return ``success'', $\irpath$
}\label{lin:avg2:6}
$j \leftarrow \argmax_{i} |\Paths_i|$, break ties by hospital id.\;\label{lin:avg2:7}
$\paths_j \leftarrow \min(\paths_j, \sum_{i \neq j} \paths_i)$\;\label{lin:avg2:8}
\ForEach{active hospital $i$}{\label{lin:avg2:9}
$\Paths_i \leftarrow$ set of up to $\paths_i$ disjoint $i$-internal paths of maximum total length.
}\label{lin:avg2:10}
from each $\Paths_i$, remove all paths of length less than $4 \, s'$\;\label{lin:avg2:remshort}
\ForEach{active hospital $i$}{\label{lin:avg2:addshort1}
    \RepTimes{$\paths_i - |\Paths_i|$}{\label{lin:avg2:addshort2}
        split a suffix of length $4 \, s'$ off the end of longest path in $\Paths_i$; add to $\Paths_i$.
    }\label{lin:avg2:addshort}
}
$\Explored \leftarrow \emptyset, \nxt \leftarrow \alpha$\;\label{lin:avg2:14}
\tikzmk{A}\Do{($|\Explored| < s'$ and can successfully set $\nxt$ to some node in an active hospital such that $x \to \nxt$ is an edge, $x \in \Explored$, and $\nxt \notin \Explored$)}{\label{lin:avg2:15}
    $\Explored \leftarrow \Explored \cup \{\nxt\}$\;\label{lin:avg2:16}
    \If{$\nxt$ lies on some $\pi \in \bigcup_{i} \Paths_i$}{\label{lin:avg2:17}
        $\sigma \leftarrow$ the suffix of $\pi$ starting at $\nxt$\;\label{lin:avg2:18}
        remove $\sigma$ from $\pi$\;\label{lin:avg2:19}
        \uIf{$|\sigma| \geq s'$}{\label{lin:avg2:20}
            \Return \hyperref[alg:stitch1]{\texttt{stitch1()}}
        } \label{lin:avg2:21} \Else {\label{lin:avg2:22}
            $\Explored \leftarrow \Explored \cup \sigma$
        }\label{lin:avg2:23}
}
}\label{lin:avg2:24}\tikzmk{B}
 \boxit{highlight}
\uIf{$|\Explored| < s'$} {\label{lin:avg2:25}
\Return ``success'', $\irpath$
}\label{lin:avg2:26}
\uElseIf{some neighbor of $v \in \Explored$ is in first $s'$ nodes of some $\pi \in \bigcup_{i} \Paths_i$}{\label{lin:avg2:27}
    \Return \hyperref[alg:stitch2]{\texttt{stitch2()}}
}\label{lin:avg2:28} \Else{\label{lin:avg2:29}
    \Return ``failure'', $\irpath$
}\label{lin:avg2:30}
\caption{High-averages mechanism}
\label[mechanism]{alg:avg}
\end{algorithm}

\begin{procedure}[!t]
\DontPrintSemicolon
$\sigma' \leftarrow$ path from $\alpha$ through $\Explored$ to $\nxt$, then $\sigma$\;\label{lin:stitch1:1}
let $\langle \pi, \pi_2, \dots, \pi_{r} \rangle$ be a hospital-alternating ordering of $\bigcup_i \Paths_i$ starting at $\pi$ such that the owner of $\pi_r$ is not the owner of $\pi$.\;\label{lin:stitch1:1b}
\uIf{$|\pi| < 2 \, s'$}{\label{lin:stitch1:2}
$\mathit{seq} \leftarrow \langle \sigma', \pi_2, \dots, \pi_r \rangle$
}\label{lin:stitch1:3}\Else{\label{lin:stitch1:4}
$\mathit{seq} \leftarrow \langle \sigma', \pi_2, \dots, \pi_r, \pi \rangle$
}\label{lin:stitch1:5}
for each $t$, try find an edge between the last $s'$ nodes in the $t$th path in $\mathit{seq}$ to the first $s'$ nodes in the $(t+1)$th path.\;\label{lin:stitch1:6}
\uIf{this stitching succeeds}{\label{lin:stitch1:7}
    \Return ``success'', stitched path
}\label{lin:stitch1:8} \Else {\label{lin:stitch1:9}
    \Return ``failure'', $\irpath$
}\label{lin:stitch1:10}
\caption{stitch1()}
\label{alg:stitch1}
\end{procedure}

\begin{procedure}[!t]
\DontPrintSemicolon
$\sigma \leftarrow$ the path from $\alpha$ through $\Explored$ to $v$\;
arrange $\bigcup_{i} \Paths_i$ in hospital-alternating sequence starting with $\pi$.\;
try to stitch the sequence of paths in order, using up to $s'$ nodes each time; add the already established stitch between $\sigma$ and $\pi$ to the front.\;
\uIf{this stitching succeeds}{
    \Return ``success'', stitched path
} \Else {
    \Return ``failure'', $\irpath$
}

\caption{stitch2()}
\label{alg:stitch2}
\end{procedure}
As we showed in the previous section, \cref{alg:s1} guarantees incentive compatibility while competing with the long-segments benchmark in terms of welfare.
While this positive result is encouraging, we would like to strengthen both the incentive properties and the efficiency of our mechanism.
Before explaining our stronger \cref{alg:avg} in more detail, we explain why these properties are needed and how we modify the original mechanism to obtain them.

\subsection{Strengthening \texorpdfstring{\cref{alg:s1}}{Mechanism 1}: Desiderata and Changes}
\paragraph{Stronger notions of incentive compatibility.}
\Cref{alg:s1} guarantees incentive compatibility, i.e., that hiding nodes cannot increase a hospital's utility from the path returned by the mechanism.
Unfortunately, it is not clear that the real path of transplantations will be what the mechanism suggests.
Indeed, at any point in the execution of the transplantation chain, a hospital might pretend that the chain ended at one of its vertices because the donor reneged, and then use the donated kidney to make more internal matches.
Even without such foul play, the lack of individual rationality in \cref{alg:s1} might disincentivize the altruist owner from participating in the mechanism.

Clearly, to satisfy individual rationality, the mechanism must be more generous to the altruist owner than in \cref{lin:s1:3} of \cref{alg:s1}, choosing a long internal path instead of just the altruist.
\Cref{alg:avg} does so, most relevantly in \cref{lin:avg2:6,lin:avg2:26}, where $\irpath$, the longest internal path starting at the altruist, is returned.

More importantly, while the requirement in \cref{alg:s1} that all paths in $\Paths_i$ for $i \neq j$ have fixed length $s'$ is convenient for analysis, it can incentivize hospitals to divert the path at the end of their last segment to extend the length of the segment.
To address this shortcoming of \cref{alg:s1}, our new mechanism must be more flexible in how it selects the $\Paths_i$, perhaps in a way that more closely resembles the redefinition of $\Paths_j$ in \cref{lin:s1:7} of \cref{alg:s1}.
However, while allowing for a wider range of $\Paths_i$, we still need to ensure that we can stitch the resulting paths, i.e., that all paths in $\upaths$ have a certain minimum length, and that no hospital has more paths than all other hospitals combined.
\Cref{alg:avg} accomplishes this by first determining a number of paths $\paths_i$ for each hospital $i$ (\cref{lin:avg2:2,lin:avg2:3,lin:avg2:7,lin:avg2:8}).
It then initializes $\Paths_i$ by selecting up to $\paths_i$ paths in a way that maximizes the total length of this set without any further constraints (\cref{lin:avg2:10}), and modifies $\Paths_i$ afterwards to make sure that its paths have an adequate minimum length and that their number is exactly $\paths_i$ (\cref{lin:avg2:remshort,lin:avg2:addshort1,lin:avg2:addshort2,lin:avg2:addshort}).
We will have to show that this always succeeds on large enough graphs, and that welfare and incentives are preserved.

\paragraph{Higher welfare.}
On the welfare side, a severe limitation of \cref{alg:s1} is that, whenever the altruist is not at the start of an internal path of length at least $s'$, the mechanism returns a trivial path.
This shortcoming is reflected in the benchmark $\sopt^s$, which is also trivial in these cases.
On the one hand, the lack of welfare in these scenarios is understandable because random edges can no longer reliably connect the altruist to long paths in other hospitals.
On the other hand, there might be deterministic cross edges that can play this role, possibly going back and forth between different hospitals for a while before reaching long paths.
If the altruist has great connectivity to paths, but these paths happen to be in other hospitals' subgraphs, we would still like our mechanism to make use of this connectivity.
The high-average benchmark $\avgopt^s$ is a way to capture this potential.
In particular, it allows the number of segments in each hospital to be arbitrarily high as long as some of these segments are long enough to make the stitching worthwhile, as we discussed in \cref{sec:lowerdiscussion}.

In contrast to \cref{alg:s1}, no path in $\bigcup_{i} \Paths_i$ has to start at the altruist in \cref{alg:avg}.
Thus, rather than stitching immediately, \cref{alg:avg} starts a graph search (from \cref{lin:avg2:14} onward), which explores nodes reachable from the altruist until it finds a node that exists on a path in the set of predetermined paths $\upaths$.
Not only is this node required to lie on a predetermined path, however; it also must also be far enough from the end of its predetermined path that, with high probability, we can stitch in another path after it.
If this search successfully finds a node on a predetermined path $\pi$, its suffix and prefix might get stitched in separately.
Since a matching-time manipulation could exploit this fact to visit a $(\paths_i + 1)$th path, all nonzero $\paths_i$ must be chosen large enough in \cref{lin:avg2:3} that a single additional path gains only lower-order amounts of utility.

Another complication is that this graph search cannot explore too many nodes because, if the initial path becomes too long, its contribution to some hospitals' utility might no longer be lower-order.
As a result, hospitals might hide nodes to force the graph search to succeed later, which might increase the length of this initial path.
To counter this, the mechanism returns with a failure in \cref{lin:avg2:30} if the search is unsuccessful after exploring more than $s'$ nodes.
Fortunately, this failure mode is unlikely to happen under truthful reports due to the random edges.
Additional cases we consider are those in which the search reaches a node on a path, but that node is too close to its end to start stitching (\cref{lin:avg2:23}), or if too few nodes are reachable from the altruist (\cref{lin:avg2:26}).

\subsection{Main Result}
We can now state our main result in full:
\begin{theorem}
\label{thm:avg}
    Let there be at least two hospitals, and let $s, p$ be parameters varying in $n$ such that $p \, s^2 \in \omega\big(\log (n)^5\big)$.
    Then, for any sequence of instances with such $n$ and $p$, running \cref{alg:avg} under truthful reports succeeds with high probability.
    In this case, call the resulting path $\pi$.
    Then, it must hold that $|\pi| \geq \big(1 - o(1)\big) \, |\avgopt^s|$ and that, by any manipulation captured by matching-time incentive compatibility, $i$ can receive at most utility $\big(1 + o(1)\big) \, |\pi \cap V_i|$.
\end{theorem}

The proof of this theorem is deferred to \cref{app:mech2proofs}.
There, we begin by going through the steps of the algorithm, which also serves to show that \cref{alg:avg} is well defined.
These observations then enable us to prove that \cref{alg:avg} succeeds with high probability~(\cref{thm:avgsucceeds}), obtains high welfare~(\cref{thm:avgwelfare}), and is matching-time incentive compatible~(\cref{thm:avgic}).
Taken together, these lemmas establish our main theorem.

\section{Alternative Modeling Assumptions}
\label{sec:assumptions}
As is often the case in the design of semi-random models, there is a range of alternatives to the modeling choices we made in \cref{sec:model}.
In this section, we discuss how some of these alternative assumptions would affect our mechanism-design problem.

For example, our model allows hospitals to hide nodes but not individual edges, a standard assumption in the existing kidney literature~\cite{AR12,TP15,AFK+15,HDH+15,BCH+17}.
In fact, however, our proofs of \cref{thm:s1ic,thm:avgic} immediately go through even if hospitals can hide not only their nodes but also internal edges and cross edges adjacent
to their nodes.
It follows that our mechanisms retain the same guarantees in this stronger model.

Additionally, we restricted our model to a single altruist and thus a single path, which simplifies the notation and makes the algorithmic problem more natural.
For appropriate generalizations of the benchmark and matching-time incentive compatibility, we believe that \cref{alg:avg,thm:avg} can be extended to multiple altruists.

Concerning the use of semi-randomness, adversarial and random steps can be combined in many ways.\footnote{This is similar to work the work of Kolla et al.~\cite{KMM11}, who show how, in the context of unique games, four different models arise from adding randomness in different phases.}
For a fixed set of nodes, fixed node ownership, a fixed altruist, and a fixed probability $p$, we can partition our semi-random model into four stages:
\begin{enumerate}
    \item \textbf{``I'':} The adversary adds internal edges.
    \item \textbf{``A'':} The adversary adds cross edges.
    \item \textbf{``R'':} Random cross edges are added with probability $p$.
    \item \textbf{``B'':} The efficiency benchmarks are computed.
\end{enumerate}
The 24 permutations of the four phases ``I'', ``A'', ``R'', and ``B'' define different semi-random models, which we will discuss exhaustively.
Note that some of these permutations define equivalent models since, if the adversarial phases ``I'' and ``A'' are adjacent, their relative order does not matter.
Furthermore, we are only interested in models where the benchmark (phase ``B'') is computed after the internal edges and at least some cross edges are inserted.

Among the remaining models, there are some (``ARIB'', ``RIAB'', ``RIBA'') in which the adversary can add internal edges adaptively to the random cross edges.
Assume that each hospital owns $\Omega(n)$ vertices and that $p \in o(1/n)$.
Then, in these semi-random models, the adversary can (with high probability) create variants of the worst-case instances in \cref{fig:worstcase}, effectively removing the benefit we get from randomness.
Indeed, with high probability, the altruist does not have random outgoing edges, only a vanishing fraction of each hospital's nodes have random edges, and there exists an isolated cross edge from the altruist owner to another hospital.
If this is the case, the adversary can build the internal paths using only nodes without random edges, and can use the isolated random cross edge to replicate the shape of the worst-case example.
As a result, any mechanism that competes with $\sopt^{c \cdot n}$ in welfare\footnote{For a small enough constant $c>0$, depending on the number of nodes per hospital.} cannot give individual rationality or incentive compatibility with high probability.
Because the argument used no adversarial cross edges, it equally applies to ``ARIB'', ``RIAB'', and ``RIBA''.

Perhaps the most interesting variant of our semi-random model is ``IRAB'', which differs from our model in that the adversarial cross edges may depend on the random cross edges.
Note that this model is more general than ours and that our lower bounds apply as a result.
Since \cref{alg:s1} does not consider the existence of cross edges at all before choosing which path to stitch, adversarial cross edges can only increase the probability of success, which means that \cref{thm:s1} continues to hold.
Unfortunately, the same cannot be said about \cref{alg:avg}, which considers cross edges in the graph search.
By adaptively adding cross edges, an adversary could guide the graph search into parts of the graph where stitching will fail.
While a version of \cref{alg:avg} without the graph search can guarantee matching-time incentive compatibility for the long-segments benchmark, we do not have results for the high-averages benchmark in this alternative model.

The two remaining models, ``IABR'' and ``IRBA'', are variants of our main model and ``IRAB'', respectively.
In these variants, the benchmark is computed not on the full compatibility graph, but on the graph with either the random or adversarial cross edges still missing.
Because the welfare attained by these benchmarks will be lower than on the full graph, positive results immediately carry over.
Since the proof of \cref{thm:irnoopt} did not make use of the fact that the optimal benchmark might increase due to random edges, this lower bound also applies to ``IABR''.

\section{Discussion}
\label{sec:discussion}
For studying the mechanism-design problem of finding long chains in kidney exchange, we proposed a semi-random model for compatibility graphs.
Since this model only assumes the existence of a very low number of random edges, this assumption can be seen as a mild regularity condition on the compatibility graph.
In this semi-random model, we developed mechanisms that limit incentives for manipulation to lower-order terms while competing with strong benchmarks in terms of welfare.

We do not claim that the mechanisms designed in this paper are ready for use in practice.
Whether they are depends on questions not yet investigated: whether the asymptotic guarantees have traction in graphs of typical size, whether hospitals can be convinced to report their pairs truthfully despite possible lower-order gains from manipulation, and whether the high-averages benchmark gets reasonably close to the optimal welfare on practical kidney-exchange graphs.
Moreover, our model abstracts from crucial aspects of the practical problem, including the existence of multiple chains and the fact that practical compatibility graphs evolve dynamically while transplantations on the chain are executed.

Thus, rather than proposing a complete solution, we hope that our work can serve as a stepping stone for future mechanisms coordinating kidney exchange with altruist-initiated chains.
Specifically, we have demonstrated that, despite the ubiquitous impossibilities of this setting, kidney exchange with donor chains is a valid target for mechanism design\,---\,as long as one is willing to trade traditional worst-case analysis for beyond worst-case methodologies.

\begin{acks}
We thank Naama Ben-David, Bailey Flanigan, Anupam Gupta, and David Wajc for helpful discussions, and we additionally thank Bailey Flanigan and the anonymous reviewers for valuable input on the draft.
This work was supported in part by the \grantsponsor{nsf}{National Science Foundation}{https://nsf.gov/} under grant \grantnum{nsf}{CCF-1733556}.
\end{acks}

\bibliographystyle{ACM-Reference-Format}
\bibliography{bibliography}

\newpage
\appendix
\section*{\LARGE Appendix}

\section{Lower Bounds for Incentive Compatibility}
\label{app:lower}
\begin{figure}[htb]
\centering
  \centering
    \includegraphics[width=.4\textwidth]{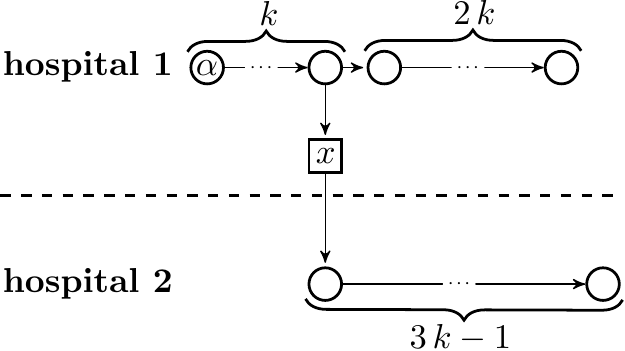}
    \caption{Without random edges, incentive compatible mechanisms cannot compete with our suboptimal benchmarks.}
    \label{fig:icworstcase}
\end{figure}
\begin{figure}[htb]
  \centering
    \includegraphics[width=.6\textwidth]{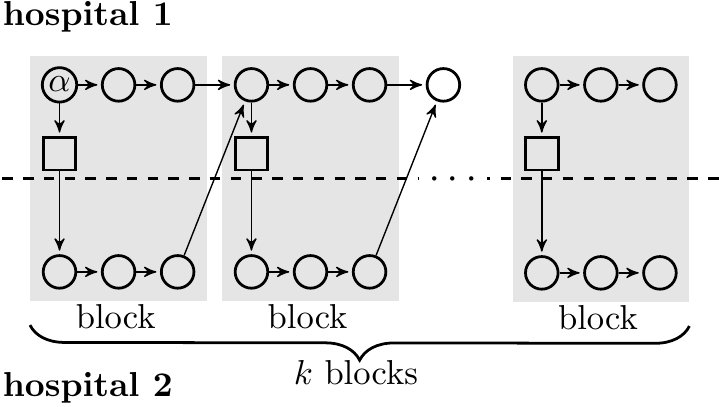}
    \caption{Even with random edges, incentive compatible mechanisms cannot compete with $\opt$.}
    \label{fig:icnoopt}
\end{figure}

\subsection{Worst-case graphs}
\label{sec:icworstcase}
\begin{theorem}
    \label{thm:icworstcase}
    Let $p=0$, i.e., suppose that no random edges are added to the worst-case base graph.
    For each deterministic mechanism, there exists a family of instances such that, on every instance, the mechanism will obtain $\Theta(n)$ less welfare than $\sopt^{n/6}$ or give some hospital $\Theta(n)$ incentive for hiding nodes.
    
    (Note that, if the mechanism falls short of $\sopt^{n/6}$, it also falls short of the stronger benchmarks $\sopt^s$ and $\avgopt^s$ for any $s \leq n/6$.)
\end{theorem}
\begin{proof}
For any $k \in \mathbb{N}_{\geq 1}$, consider the instances given in \cref{fig:icworstcase}.
Depending on the behavior of the mechanism, we will add either this instance or a slightly modified instance to the family of instances.
We will always add a graph of size in $\{6 \, k, 6 \, k - 1\}$, which means that $n \in \Theta(k)$.
Note that $\sopt^{n/6}$ selects the path of length $4 \, k$ on the full graph.

Under truthful reports, if the mechanism selects less than $2 \, k$ of hospital~2's nodes, it achieves a welfare of at most $3 \, k$ rather than the welfare of $4 \, k$ obtained by the benchmark.
Thus, the mechanism obtains a welfare that is at least $k = n/6$ lower than the one obtained by the benchmark, and we can add the full instance to the family.

Else, the mechanism selects at least $2 \, k$ nodes from hospital~2 under truthful reporting, which means that hospital~1 receives utility~$k + 1$.
Now consider the same graph with node~$x$ removed, on which $\sopt^{n/6}$ will obtain a welfare of $3 \, k$, the length of the $a$-internal path.
If the mechanism takes at most $2 \, k$ nodes when presented with this graph, its welfare is below that of the benchmark by at least $k = (n-1)/6$, and we can add the instance without $x$ to the family as above.
Else, hospital~1 receives at least $2 \, k + 1$ utility on the smaller graph, and at most $k + 1$ utility on the full graph.
Thus, hospital~1 has an incentive of $k = n/6$ to hide node $x$ when reporting its subgraph.
Thus, we add the full graph to the family.
\end{proof}

\subsection{Competing with Optimal Welfare in the Semi-Random Model}
\label{sec:icnoopt}
\begin{theorem}
    \label{thm:icnoopt}
    Fix any $p \in o(1/n)$ and any deterministic mechanism.
    Then, there exists a family of instances such that, with $\Omega(1)$ probability, the mechanism will on any instance either
    obtain $\Theta(n)$ less welfare than $\opt$ or give some hospital $\Theta(n)$ incentive for hiding nodes.
\end{theorem}
\begin{proof}
    For any $k$, consider the base graph in \cref{fig:icnoopt} with $6 \, k$ nodes.
    We will always construct a graph of size either $6 \, k$ or $5 \, k$ nodes, so $p \in o(1/k)$.
    Let $\epsilon > 0$ be a sufficiently small constant, which can be set to $1/20$.

    In each of the random graphs, the maximum welfare is at least $4 \, k$, %
    which can be obtained by not taking any random edges, and branching into hospital~2's subgraph whenever possible.
    Suppose that, over the randomness in the random edges, the mechanism has in expectation welfare below $(4 - 2 \, \epsilon) \, k$.
    Then, by Markov's inequality, the probability that welfare is above $(4 - \epsilon) \, k$ is at most $(4 - 2 \, \epsilon)/(4 - \epsilon) < 1$.
    We add this base graph to the family, and know that, with constant probability, the mechanism will achieve $\epsilon \, k \in \Theta(k) = \Theta(n)$ less welfare than the optimal path.
    
    Now consider the same base graph where all the square nodes have been removed.
    Note that this reduces $n$ by a constant factor and might change $p$, but that this $p'$ is still in $o(1/k)$.
    Since there are at most $(5 \, k)^2$ potential edges, each existing independently with probability $p'$, the distribution of $|E_p|$ is dominated by a binomial variable with $(5 \, k)^2$ trials and success probability $p'$.
    By a Chernoff bound, the probability that this binomial variable is larger than $(5 \, k)^2 \, p' + \sqrt{5 \, k} \in o(k)$ is at most $e^{-\sqrt{5 \, k}/3}$.
    Since the square nodes are removed, none of hospital~2's vertices would be reachable from the altruist without random edges, and every random edge can make at most $2$ additional vertices reachable.
    Suppose the expected welfare achieved by the mechanism on the reduced base graph is below $(3 - 2 \, \epsilon) \, k$.
    Then, by Markov, the mechanism can achieve welfare above $(3 - \epsilon) \, k$ with probability at most $(3 - 2 \, \epsilon)/(3 - \epsilon)$; with probability at least $\epsilon \, k \in \Theta(n)$, it wastes at least $\epsilon \, k$ welfare with respect to the maximal path of length at least $3 \, k$.
    Thus, in this case, the reduced base graph can be added to the family to satisfy the claim.
    Otherwise, if the expected welfare is above $(3 - 2 \, \epsilon) \, k$, at most $e^{-\sqrt{5 \, k}/3} \, 6 \, k + o(k) \, 2 \in o(k)$ of that expectation comes from nodes of hospital~2 and from any nodes in blocks with random edges.
    Thus, the expected utility for hospital~1 from non-random blocks is at least $(3 - 3 \, \epsilon) \, k$ for large enough $k$.
    Because the random variable describing this utility can never be larger than $3 \, k$ (i.e., hospital~1's total number of nodes), this utility must be larger than $8/3 \, k$ with probability at least $((3 - 3 \, \epsilon) \, k - 8/3 \, k)/(3 \, k - 8/3 \, k) = 1 - 9 \, \epsilon$.

    Return to the full base graph, knowing that the expected welfare on the whole base graph is least $(4 - 2 \, \epsilon) \, k$, and that hospital~1's utility on the reduced base graph is at least $8/3$ with probability at least $1 - 9 \, \epsilon$.
    Again, a Chernoff bound shows that there must be $o(k)$ random edges with probability at least $1 - e^{-\sqrt{6 \, k}/3}$.
    For each of block with random edges, its contribution to welfare is at most $6$, and its contribution to hospital~1's utility is at most $4$.
    All the other blocks can only be entered once, meaning that they can either provide $4$ welfare and $2$ utility for hospital~1, or $3$ welfare and $3$ utility for hospital~1.
    Let $\alpha$ be the expected fraction of blocks that do not have random edges and in which the mechanism selects the option with higher welfare.
    Then, expected welfare is at most
    \[ e^{-\sqrt{6 \, k}/3} \, 6 \, k + \alpha \, 4 \, k  + (1 - \alpha) \, 3 \, k. \]
    Since this bound also must be at least $(4 - 2 \, \epsilon) \, k$, it follows by arithmetic that $\alpha \geq 1 - 2 \, \epsilon - e^{-\sqrt{6 \, k}/3} \, 6$, which is at least $1 - 3 \, \epsilon$ for large enough $k$.
    Over the sampling of the random edges, the probability that at least two thirds of the blocks do not have random edges and are assigned the high-welfare option by the mechanism is at least $(1 - 3 \, \epsilon - 2/3)/(1 - 2/3) = 1 - 9 \, \epsilon$.
    By union bound, with probability at least $1 - 10 \, \epsilon$, it additionally holds that there are $o(k)$ random edges.
    Conditioned on these events, hospital~1's utility is at most
    \[ o(k) \, 3 + \alpha \, 2 \, k + (1 - \alpha) \, 3 \, k \leq \big(7/3 + o(1)\big) \, k.\]
    Using another union bound, with probability at least $1 - 19 \, \epsilon$, hospital~1 would receive at least $8/3 \, k$ utility by hiding all the square nodes.
    Thus, we can add the full instance to our family since, with constant probability, there is $\big(1/3 - o(1)\big) \, k \in \Theta(n)$ incentive for manipulation.
\end{proof}

\section{Proof of Lemma 3}
\label{sec:lemballs}
\lemballs*
\begin{proof}
    Imagine having $n_1$ many stacks.
    Cycle through the stacks in order and always place one ball on the current stack, going through the balls from the most prevalent color to the least prevalent one.
    In the first pass through the stacks, we place one ball of color 1 at the bottom of each stack.
    In the second pass, each stack receives a second ball because at most half of the balls have color~1.
    Note that no stack will contain two balls of the same color because no color has more than $n_1$ balls, and, thus, every color will run out before reaching the same stack a second time.
    This means that, if we append all final stacks into a cyclical list, no two balls of the same color touch.
\end{proof}

\section{Proof of Theorem 8: Analysis of Mechanism 2}
\label{app:mech2proofs}
Before showing that \cref{alg:avg} succeeds with high probability~(\cref{thm:avgsucceeds}), has high welfare~(\cref{thm:avgwelfare}), and is matching-time incentive compatible~(\cref{thm:avgic}), we start with a sequence of general observations.
These observations establish that the mechanism is well-defined and will be used in the proofs of the later lemmas.

As before, we consider $s'$ and $\log n$ as integers, ignoring rounding complications that do not influence asymptotic behavior.
We assume the input to be ``large enough'' in terms of $n$ and $s$; formally, there should be a check at the beginning of the mechanism that fails if $n$ is too low and returns $\irpath$ in this case.
Since, under truthful responses, this failure will only happen for small graphs, this does not stand in the way of the mechanism succeeding with high probability.
While hospitals could manipulate to get to this case, this is disadvantageous with an argument akin to the ones made in the proof of \cref{thm:avgic}.

Hospitals~$i$ are called \emph{active} if their $\paths_i$ is nonzero, i.e., whenever they have some internal path of length $s$.
From \cref{lin:avg2:7} on, we know that there are at least two active hospitals, one of which must be $a$.
Each active hospital $i$ must have $\paths_i \geq \log n$ at this point.
Thus, \cref{lin:avg2:8} does not make additional hospitals inactive, and it remains true that $\paths_i \geq \log n$ for all active hospitals.

In \cref{lin:avg2:10}, it must hold that the total length of each $\Paths_i$ is at least $(s/\log n) \, \paths_i$.
Indeed, consider the original values of $\paths_i$ before modification in \cref{lin:avg2:8}, which we call $\opaths_i$ here.
If $\paths_i \geq \opaths_i / \log n$, by \cref{lin:avg2:3}, there are $\opaths_i / \log n$ paths with average length $s$, which satisfy the conditions on $\Paths_i$; the total length of $\Paths_i$ can thus not be lower than $\opaths_i \, s/\log n \geq \paths_i \, s/\log n$.
Else, by choosing the $\paths_i$ longest paths from \cref{lin:avg2:3}, we obtain an average length of at least $s$, and a total length of even at least $s \, \paths_i$.

The mechanism then removes very short paths and creates additional paths long enough for stitching.
In the removal of short paths in \cref{lin:avg2:remshort}, the total length of the corresponding $\Paths_i$ might decrease by no more than $4 \, s' \, \paths_i$.
The following loop in \cref{lin:avg2:addshort1,lin:avg2:addshort2,lin:avg2:addshort} only redistributes the total length within each $\Paths_i$ without changing it.
Thus, the total length must be at least $\paths_i \, (s/\log n - 4 \, s')$ in any iteration.
Therefore, the longest path in $\Paths_i$ must have length at least $s/\log n - 4 \, s' = (1 - 4 / \log n) \, s / \log n$, which for large $n$ is much larger than $8 \, s' = (8 / \log n) \, s / \log n$.
It follows that, after the for loop has finished in \cref{lin:avg2:14}, each $\Paths_i$ has exactly $\paths_i$ disjoint paths, and that all of them have length at least $4 \, s'$.
Note that, because we adjusted $\paths_j$ in \cref{lin:avg2:8}, \cref{lem:balls} allows us to arrange all paths in a cyclic sequence such that no two adjacent paths belong to the same hospital.

In the do-while loop in \cref{lin:avg2:15,lin:avg2:16,lin:avg2:17,lin:avg2:18,lin:avg2:19,lin:avg2:20,lin:avg2:21,lin:avg2:22,lin:avg2:23,lin:avg2:24}, the mechanism performs a graph search from the altruist among the vertices of active hospitals, attempting to find a path from the altruist to one of the paths in the $\Paths_i$.
$\Explored$ is a set of explored nodes, and $\nxt$ refers to the next node to be explored.
The loop must eventually terminate because the cardinality of $\Explored$ increases in every iteration and because the loop terminates once this cardinality exceeds $s'$.
Throughout this graph search, it is maintained that the explored nodes are disjoint from all $\Paths_i$ by removing short suffixes off these paths.
Furthermore, the size of $\Explored$ can never exceed $2 \, s'$, which implies that we never reduce the length of any path below $4 \, s' - 2 \, s' = 2 \, s'$.

When a new node $\nxt$ is explored from \cref{lin:avg2:16} on, we check whether this node lies on one of the paths in the $\Paths_i$.
If it is, we check in \cref{lin:avg2:20} whether this node is at least $s'$ steps from the end of the path.
If it is not, we remove the suffix from the path and add the suffix to the explored nodes to preserve the disjointness between paths and explored nodes.
If it is far enough from the end, we attempt to stitch in \hyperref[alg:stitch1]{subprocedure \texttt{stitch1()}}.
By \cref{lem:balls}, we can find a cyclical hospital-alternating sequence of the paths in $\bigcup_i \Paths_i$, which makes the definition in \cref{lin:stitch1:1b} of \texttt{stitch1()} valid.

Note that the last $s'$ nodes of $\sigma'$ come from $\sigma$ and belong to the same hospital as $\pi$, thus not to the hospital owning $\pi_2$.
Furthermore, all later paths in the sequences defined in \cref{lin:stitch1:3,lin:stitch1:5} have minimum length $2 \, s'$.
Thus, it makes sense to search for a stitching; if successful, the stitching will be a single path starting at the altruist.
The mechanism might also fail if the stitching is unsuccessful, which we will show only happens with low probability.

Back in the main body of \cref{alg:avg}, if do-while loop terminates without returning, this might be either because at least $s'$ nodes have successfully be explored, or because less than $s'$ nodes were reachable among the active hospitals.
In the latter case, the mechanism returns in \cref{lin:avg2:26}.
Else, i.e., if at $s'$ have been explored, we check in \cref{lin:avg2:27} if some neighbor $v'$ of an explored node $v$ is at the start of some path $\pi \in \bigcup_{i} \Paths_i$.
If so, \hyperref[alg:stitch2]{subprocedure \texttt{stitch2()}} attempts to stitch $\bigcup_{i} \Paths_i$ with $\pi$ at the front.
Again, by \cref{lem:balls}, the paths can be arranged in a sequence starting at $\pi$ such that path ownership alternates.
Furthermore, every path has minimum length $2 \, s'$.
The edge from $v$ to the beginning of the path forms a stitch from a path out of the altruist to $\pi$, which we can add to the start of the stitched path if the rest of the stitching is successful.
If no such neighbor can be found, the mechanism fails in \cref{lin:avg2:30}.

\begin{lemma}
\label{thm:avgsucceeds}
    If the parameters $p$ and $s$ vary in $n$ such that $p \, s^2 \in \omega\big(\log (n)^5\big)$, \cref{alg:avg} succeeds with high probability under truthful reports.
\end{lemma}
\begin{proof}
Suppose that that the mechanism is run on the full graph.
If the mechanism fails, it might enter one of the subprocedures  \hyperref[alg:stitch1]{\texttt{stitch1()}} or \hyperref[alg:stitch2]{\texttt{stitch2()}} and fail there, or it might fail in \cref{lin:avg2:30}.
Once the mechanism enters either of the subprocedures, there are at most $\left| \bigcup_i \Paths_i \right|$ stitches between paths to be made; in each stitch, any one of $s'^2$ edges existing suffices.
Note that these edges are all cross edges and that they are disjoint from the edges previously explored by the graph search.
Thus, independently from previous steps in the mechanism, the probability of a random edge existing for a given stitch is $1 - (1 - p)^{s'^2}$,
and the existence of deterministic edges only makes this probability higher.
By a union bound over at most $n$ stitches, the probability of failure is at most $n \, (1 - p)^{s'^2} \leq e^{\log n - p \, s^2/\log (n)^4} \to 0$.

It remains to show that the mechanism is unlikely to fail in \cref{lin:avg2:30}.
Once we enter the graph search, fix paths $\pi_1 \in \Paths_{i_1}$ and $\pi_2 \in \Paths_{i_2}$ belonging to different hospitals $i_1, i_2$.
Such paths exist because at least two hospitals are active.
Whenever a node $\nxt$ is added to $\Explored$, the mechanism has, at this point in time, not queried its outgoing edges.
If the node does not belong to $i_1$, there are $s'$ possible random edges between $\nxt$ and the first $s'$ nodes of $\pi_1$.
Else, there are $s'$ possible random edges between $\nxt$ and the first $s'$ nodes of $\pi_2$.
If, at any point during the graph search, one of these edges is materialized, the mechanism can no longer fail in \cref{lin:avg2:30}.
Because, whenever the mechanism reaches \cref{lin:avg2:27}, it has explored at least $s'$ nodes, the probability of this failure mode is at most $(1 - p)^{s'^2} \leq e^{-p \, s^2 / \log (n)^4} \to 0$.
Therefore, the probability of failure anywhere in the mechanism also tends to zero.
\end{proof}

\begin{lemma}
    \label{thm:avgwelfare}
    Suppose that $s \in \omega(\log (n)^2)$.
    Whenever \cref{alg:avg} succeeds under truthful reports, it achieves welfare of at least $\big(1 - o(1)\big) \, |\avgopt^s|$.
\end{lemma}
\begin{proof}
Assuming that the mechanism does not fail, we need to show that the resulting path $\pi$ has utility at least $\big(1 - o(1)\big) \, |\avgopt^s|$.
For each hospital $i$, denote the set of $i$'s segments of $\avgopt^s$ by $\Benchmark_i$, and $|\Benchmark_i|$ by $\benchmark_i$.

From the definition in \cref{lin:avg2:3}, the original value of $\paths_i$ is clearly at least $\log n \, \benchmark_i$.
In particular, inactive hospitals cannot have any nodes in $\avgopt^s$.

If the mechanism returns in \cref{lin:avg2:6} because $a$ is inactive, there is no path satisfying the conditions on $\avgopt^s$, and there is nothing to show.
Else, if the mechanism returns in the same line because no other hospital is active, $\avgopt^s$ cannot use any other hospital's subgraph.
Therefore, selecting $\irpath$ obtains at least the welfare of the benchmark.
Finally, if the mechanism returns in \cref{lin:avg2:26}, there are fewer than $s'$ nodes reachable from the altruist within the active hospitals, who are a superset of the hospitals who might possibly have segments in a path $\avgopt^s$.
This shows that the benchmark again is trivial.
Since the returns in \cref{lin:avg2:6,lin:avg2:26} are dealt with, and since we assumed that the mechanism succeeds, it remains to consider cases in which the mechanism returns a successful stitching.

Note that, since the owners of the segments of $\avgopt^s$ alternate, $\benchmark_i \leq \sum_{i' \neq i} \benchmark_{i'} + 1$ for all $i$.
Because $\paths_i \geq \log n$ for at least two hospitals, and since $\log n \geq 2$ for large $n$, we can continue this chain as $\sum_{i' \neq i} \benchmark_{i'} + 1 \leq \sum_{i' \neq i} (\benchmark_{i'} + \paths_{i'}/2) \leq \sum_{i' \neq i} (\paths_{i'}/\log n + \paths_{i'}/2) \leq \sum_{i' \neq i} \paths_{i'}$.
Therefore, even after the modification in \cref{lin:avg2:8}, it holds for all $i$ that $\paths_i \geq \benchmark_i$.

Thus, $\Benchmark_i$ satisfies the conditions on $\Paths_i$ in \cref{lin:avg2:10}, and we know that the total length of $\Paths_i$ is at least that of $\Benchmark_i$.
It then suffices to show that, for each active hospital $i$, the utility of the path returned by the mechanism will be at least a $\big(1 - o(1)\big)$ fraction of the total length of $\Paths_i$ as defined in \cref{lin:avg2:10}.
Summing up over all active hospitals (and remembering that $\avgopt^s$ does not intersect the subgraphs of inactive hospitals), this will show that welfare is at least
\[
    \sum_{\mathclap{\text{active $i$}}} \left(1 - o(1)\right) \, \totlen{\Paths_i} \leq \sum_{\mathclap{\text{active $i$}}} \left(1 - o(1)\right) \, \totlen{\Benchmark_i} = \left(1 - o(1)\right) \, |\avgopt^s|.
\]

We now show that a fixed active hospital $i$ obtains $(1 - o(1)) \, \totlen{\Paths_i}$ utility from the mechanism.
As observed above, the removal of small paths in \cref{lin:avg2:remshort} reduces the total length of $\Paths_i$ by at most $4 \, s' \, \paths_i$.
Each stitching\,---\,whether in \hyperref[alg:stitch1]{\texttt{stitch1()}} or \hyperref[alg:stitch2]{\texttt{stitch2()}}\,---\,may lose up to $2 \, s'$ nodes for each path in $\Paths_i$ (plus possibly $s'$ more in \hyperref[alg:stitch1]{\texttt{stitch1()}} for $\sigma'$), which is another total loss of up to $s' \, (2 \, \paths_i + 1)$.
Furthermore, the graph search might cut off up to $2 \, s'$ nodes from paths in total, possibly all of them from $\Paths_i$.
Finally, the loss of up to $2 \, s'$ from possibly ignoring the prefix in \hyperref[alg:stitch1]{\texttt{stitch1()}} is already accounted for in the stitching loss.
Thus, at most $s' \, (6 \, \paths_i + 3) \leq 9 \, s' \, \paths_i$ nodes in $\Paths_i$ are not included in the final stitched path.

Since, as observed in the general remarks, the total length of $\Paths_i$ is at least $(s/\log n) \, \paths_i$, this loss expressed as a fraction of $\totlen{\Paths_i}$ is at most $(9 \, s' \, \paths_i)/\big((s/\log n) \, \paths_i\big) = 9/\log n \to 0$.
\end{proof}

\begin{lemma}
    \label{thm:avgic}
    Again, suppose that $s \in \omega(\log (n)^2)$.
    Assuming that \cref{alg:avg} succeeds under truthful reports and returns a path $\pi$, each hospital can only obtain $\big(1 + o(1)\big) \, |\pi \cap V_i|$ utility by hiding vertices, by diverting the path, or both.
\end{lemma}
\begin{proof}
Suppose again that the mechanism succeeds on the truthful graph.
If it returns in \cref{lin:avg2:6}, this is due to certain hospitals not having internal paths longer than $s$, and hiding subsets of nodes will not make the mechanism return elsewhere.
Similarly, if the mechanism returns in \cref{lin:avg2:26}, fewer than $s'$ are reachable from the altruist within the active hospitals.
If a hospital hides nodes, this can only reduce the set of active hospitals and the set nodes reachable within the active hospitals.
It is possible that hiding nodes would make the mechanism return in \cref{lin:avg2:6} instead, but then the mechanism would still return $\irpath$.
The mechanism cannot be manipulated into returning earlier in \cref{lin:avg2:21} because this would require $s'$ nodes to be reachable from the altruist within an active hospital's subgraph.
Finally, the mechanism cannot be manipulated in executing lines past \cref{lin:avg2:26} since, whenever the do-while loop ends, our assumption implies that the condition in \cref{lin:avg2:25} is true.
While hiding nodes or diverting paths can lead to a different path than the true $\irpath$ being chosen, such a manipulation cannot increase any hospital's utility.
Indeed, each such outcome is an internal path of the altruist owner starting at the altruist.
No hospital other than the altruist owner can obtain non-zero utility, and the altruist owner's utility is maximized under truthful reports.

We may thus assume that the mechanism returns a successful stitching under truthful reports.
In \cref{lin:avg2:3}, note that the definition of $\paths_i$ is monotone in the reported hospital subgraph, i.e., hiding nodes can only decrease the value.
If a hospital is not selected in \cref{lin:avg2:7} under truthful reports, no manipulation by that hospital can change that.
The hospital $j$ selected in \cref{lin:avg2:7} may hide its nodes in such a way that it is no longer selected, but, for this, the manipulated $\paths_j$ can be at most $\max_{i \neq j} \paths_i$, which is a lower bound on the truthful $\paths_j$ after modification in \cref{lin:avg2:8}.
If $j$ remains the selected hospital, it cannot increase its $\paths_j$ either, because $j$ has no control over the value $\sum_{i \neq j} \paths_i$.
Thus, even after \cref{lin:avg2:8}, $\paths_i$ is maximized by truthful reporting for each hospital $i$.
In \cref{lin:avg2:10}, $\Paths_i$ is defined such that its total length is monotone in $\paths_i$.
Thus, revealing all nodes maximizes this total length.
In the argument for efficiency, we showed that each active hospital~$i$ receives at least a $\big(1 - o(1)\big)$ fraction of this total length as their utility.
It suffices to show that no manipulation, including diverting the path, can obtain more utility than a $\big(1 + o(1)\big)$ fraction of the total length of $\Paths_i$ at the state of its definition in \cref{lin:avg2:10} under truthful reports.

To begin, if a manipulation makes the mechanism return $\irpath$ in \cref{lin:avg2:6}, in \cref{lin:avg2:26}, or in one of the failure states, no hospital $i$ can receive more utility than the total length of $\Paths_i$.
Indeed, all hospitals except for the altruist owner obtain zero utility from this outcome (and cannot divert the path).
Since the altruist owner $a$ is active, the singleton set with $\irpath$ (or with $\irpath$ diverted at some point) would have been a valid choice for $\Paths_a$, which means that the total length of $\Paths_a$ must be at least the utility obtained in this scenario.

Thus, we can restrict our focus on manipulations that prompt the mechanism to return a successful stitching.
We claim that hospital $i$'s nodes that end up on the path (possibly diverted) are always contained in $\paths_i + 1$ many disjoint $i$-internal paths, except for up to $2 \, s'$ additional nodes.
If this is indeed true, under truthful reporting, $\Paths_i$ could have been set to the $\paths_i$ longest of these paths because the truthful $\paths_i$ is at least as large as that under some manipulation.
Thus, if the total length of the truthful $\Paths_i$ is $\ell$ and the utility of the path resulting from the deviation is $d$, it must hold that $\ell \geq (d - 2 \, s') \, \paths_i/(\paths_i + 1)$, i.e., that $d \leq (1 + 1/\paths_i) \, \ell + 2 \, s'$.
Since $\ell \geq (s/\log n) \, \paths_i$, since $\paths_i \geq \log n$, and thus $\ell \geq s$, the utility of the deviation is at most $\ell \, (1 + 1 / \log n + 2 / \log (n)^2) \in \big(1 + o(1)\big) \, \ell$, as claimed. 

We will now show this remaining claim, i.e., that all but up to $2 \, s'$ of hospital~$i$'s nodes on a successful stitching are contained in up to $\paths_i + 1$ disjoint paths.
Without diverting paths, the path produced by the mechanism is always a subset of $\bigcup_{i} \Paths_i$ plus up to $2 \, s'$ nodes to reach these paths, which satisfies the claim with even only $\paths_i$ disjoint paths.

We can thus concentrate on manipulations that divert the path.
In \hyperref[alg:stitch2]{\texttt{stitch2()}}, and in \hyperref[alg:stitch1]{\texttt{stitch1()}} if the path prefix is small (\cref{lin:stitch1:2}), hospital~$i$ will receive up to $2 \, s'$ nodes on the path out of the altruist found by the graph search, and after this exactly $\paths_i$ stitched segments.
Diverting the path at any point modifies one of these segments and removes the later segments.
Thus, taking the segments as up to $\paths_i$ disjoint internal paths, the claim is satisfied.

It remains to consider the case where the prefix is long (\cref{lin:stitch1:4}) in \hyperref[alg:stitch1]{\texttt{stitch1()}}.
The path returned by the mechanism will have up to $2 \, s'$ nodes in the beginning and then $\paths_i + 1$ segments.
While the segments $\pi$ and $\sigma$ (which is the remainder of $\sigma'$ after the nodes in the beginning) lie on a single path, the hospital could divert the path to change this.
In particular, it could divert $\pi$ such that it no longer lies on a single path with $\sigma$.
But even then, there are at most $\paths_i + 1$ segments, which shows the claim.
\end{proof}
\end{document}